%% file: joins.tex
\documentclass[sigconf]{acmart}

\acmISBN{}

\usepackage{listings}
\usepackage{xspace}
\usepackage[dvipsnames]{xcolor}
\usepackage[skip=0pt]{subcaption}
\usepackage{graphicx}
\usepackage{caption}
\usepackage{array, makecell}
\usepackage{pgfplots}
\usepackage[ruled,vlined,linesnumbered]{algorithm2e} 
\usepackage{enumitem}
\usepackage{amsthm}
\usepackage{microtype}
\usepackage{booktabs}
\usepackage{hyperref}
\usepackage{tcolorbox}

\hypersetup{
    colorlinks,
    linkcolor={red!50!black},
    citecolor={blue!50!black},
    urlcolor={blue!80!black}
}
\def\BibTeX{{\rm B\kern-.05em{\sc i\kern-.025em b}\kern-.08em
    T\kern-.1667em\lower.7ex\hbox{E}\kern-.125emX}}
  
\begin{document}

\title{Accelerating Approximate Analytical Join Queries over Unstructured Data 
with Statistical Guarantees}

\author{Yuxuan Zhu}
\affiliation{%
  \institution{University of Illinois Urbana-Champaign}
  \city{Urbana}
  \country{USA}
}
\email{yxx404@illinois.edu}

\author{Tengjun Jin}
\affiliation{%
  \institution{University of Illinois Urbana-Champaign}
  \city{Urbana}
  \country{USA}
}
\email{tengjun2@illinois.edu}

\author{Chenghao Mo}
\affiliation{%
\institution{University of Illinois Urbana-Champaign}
\city{Urbana}
\country{USA}
}
\email{cmo8@illinois.edu}

\author{Daniel Kang}
\affiliation{%
\institution{University of Illinois Urbana-Champaign}
\city{Urbana}
\country{USA}
}
\email{ddkang@illinois.edu}

\input{macros.tex}

\input{tex/abstract.tex}

\maketitle

\input{tex/intro}
\input{tex/query.tex}
\input{tex/use_case.tex}
\input{tex/background.tex}
\input{tex/algorithms.tex}
\input{tex/theories.tex}
\input{tex/eval}
\input{tex/related_work}
\input{tex/conclusion}

\bibliographystyle{ACM-Reference-Format}
\bibliography{joins}

\cleardoublepage
\appendix
\input{tex/appendix/cost_analysis.tex}

\section{Theoretical Justification}

In this section, we prove that \bas converges to the optimal allocation at the 
rate of $\mathcal{O}(1/\sqrt{b_1})$ and asymptotically outperforms or matches 
the standalone sampling algorithm. In addition, we prove that \bas be extended to accelerate approximate selection join queries with recall guarantees.

\input{tex/appendix/notations.tex}
\input{tex/appendix/setup.tex}
\input{tex/appendix/optimal.tex}

\input{tex/appendix/comparison.tex}

\input{tex/appendix/selection.tex}
\input{tex/appendix/bootstrap.tex}
\input{tex/appendix/query_semantics.tex}

\end{document}

%% file: macros.tex
\newcommand*\LSTfont{\small\ttfamily\SetTracking{encoding=*}{-60}\lsstyle}
\definecolor{defaultcolor}{HTML}{1f77b4}
\newcommand{\syntaxhighlight}[1]{\textbf{\textcolor{defaultcolor}{#1}}}
\lstset{
  frame=tb,
  language=SQL,
  basicstyle=\footnotesize\ttfamily,
  keywordstyle=\syntaxhighlight,
  commentstyle=\textcolor{gray},
  breaklines=true,
  sensitive=true
}

\newcommand\join{\texttt{JOIN}\xspace}
\newcommand{\todo}[1]{\textcolor{red}{#1}}
\newcommand{\mathbbm}[1]{\text{\usefont{U}{bbm}{m}{n}#1}} 
\newcommand{\revise}[1]{\textcolor{blue}{#1}}
\newcommand\ie{i.e.}
\newcommand\eg{e.g.}
\newcommand\bas{\textsc{BaS}\xspace}
\newcommand\name{\textsc{JoinML}\xspace}
\newcommand\wwj{\textsc{WWJ}\xspace}
\newcommand\block{\textsc{Blocking}\xspace}
\newcommand\abae{\textsc{Abae}\xspace}
\newcommand\blazeit{\textsc{BlazeIt}\xspace}
\newcommand\uniform{\textsc{Uniform}\xspace}
\newcommand\joinmlbas{\textsc{JoinML-BaS}\xspace}
\newcommand\joinmlwwj{\textsc{JoinML-WWJ}\xspace}
\newcommand\joinmlblock{\textsc{JoinML-block}\xspace}
\newcommand\namenoci{\textsc{Block-NoCI}\xspace}
\newcommand\namefixci{\textsc{JoinML-Fixed}\xspace}
\newcommand\nameoptci{\textsc{JoinML-Adapt}\xspace}
\newcommand\nameis{\textsc{JoinML-IS}\xspace}
\newcommand\aavg{\texttt{AVG}\xspace}
\newcommand\acount{\texttt{COUNT}\xspace}
\newcommand\asum{\texttt{SUM}\xspace}
\newcommand\amin{\texttt{MIN}\xspace}
\newcommand\amax{\texttt{MAX}\xspace}
\newcommand\amedian{\texttt{MEDIAN}\xspace}
\newcommand\atopk{\texttt{TopK}\xspace}
\newcommand\agroupby{\texttt{GroupBy}\xspace}
\newcommand\dquora{Quora\xspace}
\newcommand\dcompany{Company\xspace}
\newcommand\dwebmasters{Webmasters\xspace}
\newcommand\dcityhuman{CityHuman\xspace}
\newcommand\droxford{Roxford\xspace}
\newcommand\dflickr{Flickr30K\xspace}
\newcommand\dveri{VeRi\xspace}
\newcommand{\review}[1]{}

\newcommand{\minihead}[1]{{\vspace{.2em}\noindent\textbf{#1.} }}
\newenvironment{denseitemize}{
\begin{itemize}[topsep=2pt, partopsep=0pt, leftmargin=1.5em]
  \setlength{\itemsep}{2pt}
  \setlength{\parskip}{0pt}
  \setlength{\parsep}{0pt}
}{\end{itemize}}

%% file: tex/abstract.tex
\begin{abstract}
Analytical join queries over unstructured data are increasingly prevalent in 
data analytics. Applying machine learning (ML) models to label every pair 
in the cross product of tables can achieve state-of-the-art accuracy, but 
the cost of pairwise execution of ML models is prohibitive. Existing algorithms, 
such as embedding-based blocking and sampling, aim to reduce this cost. 
However, they either fail to provide statistical guarantees (leading to errors 
up to 79\% higher than expected) or become as inefficient as uniform sampling.

We propose blocking-augmented sampling (\bas), which simultaneously achieves 
statistical guarantees and high efficiency. \bas optimally orchestrates 
embedding-based blocking and sampling to mitigate their respective 
limitations. Specifically, \bas allocates data tuples in the cross product into two 
regimes based on the failure modes of embeddings. In the regime of 
false negatives, \bas uses sampling to estimate the result. In the regime of 
false positives, \bas applies embedding-based blocking to improve 
efficiency. To minimize the estimation error given a budget for ML executions, 
we design a novel two-stage algorithm that adaptively allocates the 
budget between blocking and sampling. Theoretically, we prove that \bas 
asymptotically outperforms or matches standalone sampling. On real-world 
datasets across different modalities, we show that \bas provides valid 
confidence intervals and reduces estimation errors by up to 19$\times$, compared
to state-of-the-art baselines.
\end{abstract}

%% file: tex/intro.tex
\section{Introduction}\label{sec:intro}
The increasing capability of machine learning (ML) and large language models 
(LLMs) allows data analysts to analyze semantically related objects across 
multiple unstructured datasets automatically \cite{crossencoder-1,
crossencoder-2,crossencoder-3,veri1,veri2,model-osnet,model-transreid,
model-blip,clip}. For example, a business investor can analyze differences in 
stock prices among related companies \cite{usecase-company} by computing 
the \aavg over a join of two sets of company records with a join condition based 
on the company description.
\begin{lstlisting}[keywords={SELECT,AVG,FROM,JOIN,ON,NL}]
SELECT AVG(c1.price - c2.price) FROM companies AS c1 
JOIN companies AS c2 ON NL(`Company: {c1.description} and Company: {c2.description} build similar products')
\end{lstlisting}

To accurately process queries whose join conditions involve arbitrary 
semantics, an analyst can manually label or apply an ML model to each pair of 
records. This pairwise evaluation method has demonstrated high accuracy across 
various applications, including small language models for entity 
matching \cite{dataset-company-work2,paganelli2022analyzing,teong2020schema} 
and text retrieval \cite{crossencoder-3,crossencoder-1,crossencoder-2,
zhuang2023rankt5}, vision-language models for image-text retrieval 
\cite{model-blip,vendrow2025inquire,li2021align}, and large language models for 
challenging data wrangling tasks \cite{em-llm1,peeters2023entity,wang2024match,
patel2024lotus,liu2024optimizing}. We refer to this pairwise 
evaluation method as the \textit{Oracle}. Unfortunately, the Oracle 
method can be prohibitively expensive since the number of ML model invocations 
grows at least quadratically with the number of data records and 
exponentially with the number of tables to be joined.\footnote[1]{For 
example, running the Oracle on an entity matching dataset, the \dcompany 
dataset from the Magellan benchmark \cite{usecase-em}, costs \$709K with GPT-4o 
or \$43K with GPT-4o mini. The pairwise evaluation results in 567 million 
tokens. We estimated the monetary cost according to 
the state-of-the-art prompt \cite{em-llm1} and the pricing from 
OpenAI \cite{openai-pricing}.}

To mitigate these high costs, researchers have proposed two approximation 
methods that reduce the number of ML model invocations: 
\textit{embedding-based blocking} and \textit{sampling}. However, when 
used to process analytical join queries over unstructured data, 
embedding-based blocking can lead to arbitrary errors, while sampling 
is inefficient.

\minihead{Embedding-based Blocking} In efficient entity matching systems, 
embedding-based blocking first uses embedding similarities to filter 
out data pairs and then applies ML models for pairwise evaluation 
\cite{blocking2,dataset-company-work1,dataset-company-work2,deep-em}. In 
contrast to the Oracle method that applies ML models to each data 
pair, embedding-based blocking applies embedding models on each data record and
only uses the Oracle on a limited number of data pairs, which effectively 
reduces the cost of executing expensive ML models.

However, embedding models can be unreliable and are not guaranteed to 
precisely capture the semantics of interest 
\cite{muennighoff2023mteb,thakurbeir,huang2023survey,li2023one}. 
Consequently, two data records that satisfy the join condition may have 
a low embedding similarity (\ie, false negatives), causing embedding-based 
blocking to filter out such a pair of records. In this case, when used to compute 
data aggregations, embedding-based blocking skews aggregation results toward 
data with high embedding similarity. Thus, 
embedding-based blocking does not provide statistical guarantees. Empirically, 
we find that blocking can result in errors up to 79\% higher than expected even 
when the threshold of filtering is calibrated on a validation dataset (\S 
\ref{subsec:stats-analysis}).

\minihead{Sampling} To achieve statistical guarantees, researchers have used 
sampling techniques to process analytical join queries with approximate 
semantics \cite{structured-join-synopsis,correlated-sampling,ripplejoin}. With 
independently and identically distributed samples, we can calculate confidence 
intervals (CIs) to provide statistical guarantees. However, due to the sparsity of 
positive pairs (Figure \ref{fig:datasets}, \S \ref{subsec:setup}), uniform 
sampling can be inefficient, requiring a large sample size to achieve a
small error. In Section \ref{subsec:e2e}, we show that achieving an average error
of 5\% for the \dcompany dataset requires uniformly sampling $1.4\times10^7$ 
data pairs and processing them with the Oracle
(Figure~\ref{fig:rrmse-company}).\footnote[2]{On average, $1.4\times10^7$ pairs 
contains 20B tokens. Processing them would cost \$26K with GPT-4o or \$1.6K with 
GPT-4o-mini.} 

We propose \textit{Weighted Wander Join} (\wwj) to improve the efficiency of 
sampling using embeddings. \wwj extends Wander Join 
\cite{structured-join-index}, an efficient join algorithm for structured data, 
with an approximate index computed from embeddings. Intuitively, \wwj 
follows the idea of importance sampling to reduce sampling variance by 
upweighting the probability of sampling data pairs with high embedding
similarities \cite{kloek1978bayesian}. Nevertheless, the 
effectiveness of \wwj depends heavily on the precision of the embeddings. We 
empirically show that \wwj becomes as inefficient as uniform sampling, if the 
embedding similarities contain many false positives (Figure \ref{fig:rrmse}, 
\S \ref{subsec:e2e}).

\minihead{Blocking-augmented Sampling (\bas)}
To address the limitations of standard statistical methods over 
imperfect embeddings, we propose \textit{Blocking-augmented Sampling}. 
Our core novelty lies in the dynamic orchestration of two techniques 
based on the specific failure modes of embedding models: false negatives 
(which break blocking) and false positives (which break importance sampling).

\bas addresses this by treating the blocking threshold not as a static 
filter, but as a dynamic decision boundary that splits the search space into two 
regimes: (1) a blocking regime (high similarity, high false positive risk), where we 
deterministic execute the Oracle to eliminate variance caused by frequent 
non-matches; (2) a sampling regime (low similarity, high false negative risk), 
where we apply importance sampling to correct the bias introduced by blocking. 
Unlike standard stratified sampling where strata are fixed, \bas solves a 
complex optimization problem to determine the optimal boundary between these 
regimes.


Given a fixed budget for Oracle executions, it is challenging to 
optimally allocate the data pairs into two regimes and assign the budget 
between blocking and \wwj. Empirically, the optimal allocation that minimizes 
the estimation error varies significantly across different datasets. Using the 
optimal allocation can result in estimation errors that are up to 99\% smaller 
than the worst allocation (Figure \ref{fig:optimal}, \S \ref{subsec:optimality}). 
Theoretically, we show that the optimal budget allocation depends on the 
sampling variance of \wwj on different parts of the dataset. Unfortunately, the 
sampling variance is unknown unless the Oracle is executed on the entire 
dataset.

To automatically determine the optimal allocation, we propose an adaptive 
allocation algorithm via two-stage stratified sampling. Initially, \bas 
stratifies data tuples based on the embedding similarity. In the first stage, 
\bas applies pilot sampling \cite{pilot-sampling} to estimate the sampling variance of each stratum. 
Based on these estimates, \bas determines the optimal allocation that minimizes 
the overall estimated sampling variance. In the second stage, \bas executes \wwj 
and blocking according to the allocated regime and budget. We prove that \bas 
not only converges to the optimal budget allocation but also outperforms or 
matches standalone \wwj asymptotically. Finally, to obtain valid CIs as 
statistical guarantees, \bas adopts bootstrapping that estimates statistical 
distributions via resampling \cite{boostrap-comparison,pol2005relational}. We 
theoretically justify the validity and empirically 
demonstrate the coverage of the derived CIs.

We implement \wwj and \bas in a prototype system called \name to accelerate 
approximate analytical join queries over unstructured data. We evaluate \name 
using six real-world datasets and two synthetic datasets, including 
text, image, and bimodal data. Our results reveal that \bas reduces the root mean squared error by up to 19$\times$ 
compared to the uniform sampling, blocking and prior systems for semantic 
operators. We also show that the CIs provided by \bas achieve the expected 
coverage. In additional to analytical queries, we show that \bas can improve the 
precision of selection queries over unstructured data by up to 69\%, compared to 
baselines. In our ablation study, we show that the allocation algorithm closely 
approaches optimal performance. Our contributions are as follows:
\begin{enumerate}[leftmargin=*]
    \item We develop \name algorithms, \wwj and \bas, that can accurately and 
    efficiently answer approximate analytical join queries with statistical 
    guarantees.
    \item We prove the convergence rate of \bas to the optimal allocation and 
    that \bas outperforms or matches standalone sampling asymptotically.
    \item We implement and evaluate \name in real-world text, images, and 
    bimodal datasets, reducing root mean squared errors by up to 19$\times$ compared to state-of-the-art 
    baselines.
\end{enumerate}

%% file: tex/query.tex
\section{Query Syntax and Semantics}
\begin{figure}[t]
\centering
\lstset{language=SQL}
\begin{lstlisting}[keywords={SELECT,AVG,FROM,JOIN,ON,NL,SUM,COUNT,AND,ORACLE,BUDGET,WITH,PROBABILITY}]
SELECT {AVG | SUM | COUNT} {field | EXPR(field)}
FROM table_name1 JOIN table_name2 JOIN ...
ON NL(join_condition) AND ...
ORACLE BUDGET b WITH PROBABILITY p
\end{lstlisting}
\caption{Query syntax of \name.}
\label{fig:syntax}
\end{figure}

\name accelerates analytical queries involving joins and similarity-based join 
conditions. Currently, we support linear aggregates, including \acount, \asum, 
and \aavg. We illustrate the query syntax in Figure \ref{fig:syntax}. Besides a 
SQL query, users can specify an execution budget of Oracle and the coverage probability 
for the CI. To understand the semantics of a \name query, we introduce the 
concepts of Oracle and embedding similarity.

\minihead{Oracle} We define Oracle as a method that operates on a pair of 
records to determine if they satisfy the join condition. We assume that the 
Oracle is the state-of-the-art method that outputs the ground-truth label. For 
example, the Oracle method can be a foundation model that decides if two 
sentences contain the same entity \cite{em-llm1}, a fine-tuned model that 
accesses whether two images show the same car \cite{veri1,veri2}, or a human 
labeler for particularly challenging tasks 
\cite{ziems2024can,demszky2020learning,pilehvar2018wic}. 

\minihead{Embedding Similarity} We define the similarity between two 
records $t_1$ and $t_2$ as the cosine similarity of their embedding vectors:
\begin{equation*}
    \mathrm{sim}\left(E(t_1), E(t_2)\right) 
    = \frac{E(t_1)^\top  \cdot E(t_2)}{|E(t_1)|\cdot|E(t_2)|}
\end{equation*}
The embedding vectors $E(\cdot)$ can be obtained by applying a pre-trained 
embedding model on each data record individually. Therefore, compared to Oracle 
methods that are executed on each pair of records, using embedding similarity 
incurs significantly lower computational costs. However, since the embedding 
model does not consider the joint semantic across data records, the embedding 
similarity can be less accurate than the Oracle.

\minihead{Query Semantics} A \name query will output an estimate $\hat\mu$ of the 
answer and a CI $[l, u]$. Suppose $\mu$ is the result of processing the query 
exhaustively via the Oracle. With a budget $b$ and a probability $p$, we provide 
the following guarantees: (1) The Oracle will not be executed on more than $b$ 
pairs; (2) $\mathbb{P}\left[l \le \mu \le u\right] \ge p$. Our query semantics 
and guarantees are similar to prior AQP over unstructured data without joins 
\cite{blazeit,supg,abae}.

%% file: tex/use_case.tex
\section{Use Cases} \label{sec:use-cases}
Analytical queries with joins are increasingly prevalent in analytics over 
unstructured data, such as plagiarism analysis \cite{usecase-plagiarism,
usecase-plagiarism1,usecase-plagiarism2}, business analysis \cite{usecase-company,
usecase-company1}, and traffic condition monitoring 
\cite{application-vehicle-reid,application-vehicle-reid1}. In addition to the
business analysis query introduced in Section \ref{sec:intro}, we 
demonstrate two more typical use cases. In the example queries, we can use an 
oracle budget of 1,000,000 and a confidence of 0.95.

\minihead{Plagiarism Analysis} 
Sentence-level analysis is a common approach in plagiarism detection that 
involves comparing each sentence of an input article with every sentence in a 
reference collection \cite{usecase-plagiarism,white2004sentence}. The plagiarism 
score of the article can be calculated as the percentage of sentences that are 
paraphrased from the reference collection. To compute the percentage, we need to 
join the sentences of the article with those in the reference collection and 
count the number of sentence pairs that are considered paraphrased:
\begin{lstlisting}[morekeywords={SELECT,AVG,FROM,JOIN,ON,NL},sensitive=true]
SELECT COUNT(*) FROM article JOIN db ON NL(`{article.sentence} is paraphrased from {db.sentence}.')
\end{lstlisting}


\minihead{Traffic Analysis} 
An urban planner may be interested in analyzing the average traffic time across 
a road segment by calculating the time for the same vehicle to travel from one 
end to the other \cite{application-vehicle-reid}. To find the time difference, 
the urban planner needs to join two videos and identify common vehicles:
\begin{lstlisting}[keywords={SELECT,AVG,FROM,JOIN,ON,NL}]
SELECT AVG(video1.ts - video2.ts) FROM video1 JOIN video2 ON NL(`Frame {video1.frame} and Frame {video2.frame} contains the same car.')
\end{lstlisting}

\minihead{Join Order Optimization} 
Optimizing multi-way semantic joins requires accurate cardinality estimates to 
determine efficient execution orders. However, for unstructured predicates 
requiring ML evaluation, deriving these statistics is prohibitively expensive. 
\bas resolves this by efficiently estimating \acount with tight confidence 
intervals using a minimal budget. These estimates allow the query optimizer to 
reliably select the optimal join order, preventing performance regressions 
caused by poor planning without requiring full pairwise evaluation (Section \ref{sec:eval-app}).

%% file: tex/background.tex
\section{Background} \label{sec:background}
Existing algorithms in approximate query processing (AQP) and entity matching 
(EM) can be applied to accelerate analytical join queries over unstructured data. 
In this section, we introduce two typical algorithms: sampling and blocking.

\begin{figure}[t!]
    \centering
    \includegraphics[width=0.8\linewidth]{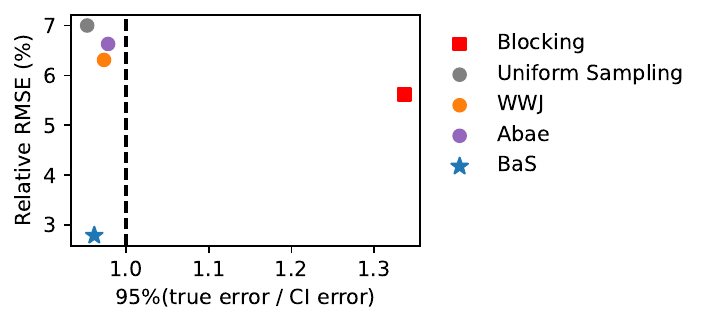}
    \caption{Given a \acount query with an Oracle budget of 5,000,000 and a 
    probability of 95\%, sampling algorithms leads to high estimation error while 
    blocking fails to achieve statistical guarantees. Achieving statistical
     guarantees means that the 95th percentile of the true error is less than the 
    error bounded by a 95\% CI (to the left of the dashed line).}
    \label{fig:background}
\end{figure}

\subsection{Sampling} \label{subsec:bg-sampling}
Sampling techniques are widely used in AQP to reduce 
query costs while providing CIs as statistical guarantees. Prior work has developed 
efficient sampling algorithms for analytical join queries with join keys, 
such as correlated sampling \cite{correlated-sampling}, universe sampling 
\cite{structured-join-sampling}, and Wander Join \cite{structured-join-index}. 
We can apply these algorithms to accelerate joins over unstructured data.

\minihead{Limitations}
Existing efficient sampling algorithms for joins over structured data requires 
exact indices or hash functions on the join keys. However, there is no join key 
for semantic join queries over unstructured data. To obtain semantic indices or 
hash functions, one way is applying the Oracle method across all
data tuples, which is prohibitively expensive. Without indices or hash 
functions, we can only use uniform sampling, which is inefficient because 
positive pairs are often sparse in the cross product of tables. Empirically, 
given a fixed error budget, uniform sampling can lead to 2.6$\times$ higher 
error than our proposed method (Figure \ref{fig:background}).

\begin{algorithm}[t]

    \DontPrintSemicolon
    \SetKwInOut{Input}{Input}\SetKwInOut{Output}{Output}
    \SetKwComment{SideComment}{\textcolor{blue}{//~}}{}
    \SetKwComment{CommentStyle}{\textit{// }}{}
    \SetKwComment{LineComment}{\textcolor{blue}{\textit{/* }}}{\textcolor{blue}{  */}}
    \SetFuncSty{text}
    \SetKwFunction{FCI}{StandardCI}
    \SetKwFunction{FSample}{Sample}
    \SetKwProg{Fn}{Function}{:}{}

    \Input{input tables $T_1, T_2$, Oracle $O(\cdot)$,
    embedding $E(\cdot)$, attribute $g(\cdot)$, embedding similarity 
    $\mathrm{sim}(\cdot, \cdot)$, similarity score threshold $\tau$}

    $S \leftarrow \emptyset$\;
    \For{each data record $t_1$\ \text{in}\  $T_1$}{
        \For{each data record $t_2$\ \text{in}\ $T_2$}{
            $\tau_i = \mathrm{sim}\left(E(t_1), E(t_2)\right)$\;
        }
        \If{$\tau_i \ge \tau$}{
            $S \leftarrow S \cup \{(t_1, t_2)\}$ \tcp{\textcolor{blue}{Forming data blocks}}
        }
    }
    $X \leftarrow \{O(s) \cdot g(s): s \in S\}$\;
    \Return{$\sum_{i=1}^{|X|}X_i$}\;

    
    \caption{Emedding-based blocking for a \asum aggregation over a join of two tables.}
    \label{alg:blocking-naive}
    
\end{algorithm}

\subsection{Embedding-based Blocking} \label{subsec:bg-blocking}
To reduce the number of pairwise comparisons, traditional blocking 
algorithms typically uses a blocking phase to prune the search space 
\cite{blocking1,li2020survey,christophides2015entity,
fellegi1969theory,em-expert2,em-expert5,bilenko2006adaptive,
michelson2006learning,papadakis2012blocking}. While traditional blocking 
algorithms cluster records into disjoint buckets (blocks), modern deep learning 
approaches use embedding-based blocking \cite{blocking1,blocking2,
dataset-company-work1,dataset-company-work2}. Throughout this paper, blocking 
refers to this embedding-based filtering with thresholds. It implicitly defines 
a block of candidate pairs as the set of tuples whose embedding similarities 
exceed a threshold $\tau$. By filtering out pairs below $\tau$ (Line 5 in 
Alg.~\ref{alg:blocking-naive} and \ref{alg:blocking-noci}), the algorithm 
effectively blocks unrelated data from expensive pairwise comparison.
We show the procedure for estimating \asum in Alg.~\ref{alg:blocking-naive}, 
where the Oracle is executed only on the candidate block.

In practice, we often have a budget for Oracle invocations and may need
to join more than two tables. To address these, we present a direct extension of 
the original embedding-based blocking as a baseline in Algorithm 
\ref{alg:blocking-noci}. As shown, to satisfy the predefined Oracle budget, we 
perform random sampling over the resulting set of data tuples after filtering 
and then estimate the target statistic using the sample. We obtain a point 
estimate and a CI via a standard approach. Furthermore, to support $k$-table 
joins ($k \ge 2$) in general, we calculate the similarity of a $k$-tuple as the 
joint probability that the $k$ tuples match, assuming embedding similarity 
approximates the probability that two data records match.

\minihead{Limitations} Unfortunately, the predefined threshold $\tau$ does not 
guarantee zero false negatives on the evaluation dataset, even when we calibrate
it on a validation dataset. Therefore, there is no guarantee on the error of 
the point estimate or the validity of the CI. Empirically, we show that 
blocking can result in true errors higher than the errors bounded by the 
provided CIs (Figure \ref{fig:background}).

\input{tex/figures/alg_blocking.tex}

%% file: tex/figures/alg_blocking.tex
\begin{algorithm}[t]
    \DontPrintSemicolon
    \SetKwInOut{Input}{Input}\SetKwInOut{Output}{Output}
    \SetKwComment{SideComment}{\textcolor{blue}{//~}}{}
    \SetKwComment{CommentStyle}{\textit{// }}{}
    \SetKwComment{LineComment}{\textcolor{blue}{\textit{/* }}}{\textcolor{blue}{  */}}
    \SetFuncSty{text}
    \SetKwFunction{FCI}{StandardCI}
    \SetKwFunction{FSample}{Sample}
    \SetKwProg{Fn}{Function}{:}{}

    \Input{input tables $T_1, \ldots, T_k$, Oracle $O(\cdot)$,
    embedding $E(\cdot)$, attribute $g(\cdot)$, embedding similarity 
    $\mathrm{sim}(\cdot, \cdot)$, similarity score threshold $\tau$, oracle budget $b$, probability $p$}

    $D \leftarrow \mathrm{CrossProduct}(T_1, \ldots, T_k)$\;
    $S \leftarrow \emptyset$\;
    \For{each $k$-tuple $\left(t_1, \ldots, t_k\right)$\ \text{in}\  $D$}{
        $\tau_i \leftarrow \prod_{j=1}^{k-1} \mathrm{sim}\left(E(t_{j}), E(t_{j+1})\right)$\;
        \If{$\tau_i \ge \tau$}{
            $S \leftarrow S \cup \{\left(t_1, \ldots, t_k\right)\}$ \tcp{\textcolor{blue}{Forming data blocks}}
        }
    }
    \If{$|S| \le b$}{
        $X \leftarrow \{O(s) \cdot g(s): s \in S\}$\;
        \Return{$\sum_{i=1}^{|X|}X_i$}\;
    }
    \Else{
        $\tilde{S} \leftarrow \FSample(S, b)$\;
        $X \leftarrow \{O(s)\cdot g(s): s \in \tilde{S}\}$\;
        \Return{\FCI{$X, p, |S|$}}
    }
    \Fn{\FCI{$X, p, N$}}{
        $\hat\mu \leftarrow \frac{1}{|X|}\sum_{i=1}^{|X|}X_i$\; 
        $\hat\sigma \leftarrow \frac{1}{|X|-1}\sum_{i=1}^{|X|}\left(X_i - \hat\mu\right)^2$\;
        $l, u \leftarrow \hat\mu \mp z_{(1+p)/2}\cdot \frac{\hat\sigma}{\sqrt{|X|}}$\;
        \Return{$N\hat\mu, Nl, Nu$}\;
    }
    
    \caption{Embedding-based blocking for \asum with a predefined Oracle budget.}
    \label{alg:blocking-noci}
\end{algorithm}

%% file: tex/algorithms.tex
\section{Efficient \name Algorithms with Guarantees} \label{sec:alg}

In this section, we present our novel \name algorithms for analytical join 
queries over unstructured data which overcomes the limitations of baselines, 
simutaneously achieving statistical guarantees and high efficiency. Our key 
insights are:
\begin{enumerate}[leftmargin=*]
\item The sparsity of positive tuples in the cross product of tables makes 
      uniform sampling inefficient. To address it, we design a 
      non-uniform sampling algorithm (\ie, Weighted Wander Join) with the 
      embedding as the approximate index (\S \ref{subsec:wwj}).
\item False positives of the embedding similarity compromise the efficiency of 
      \wwj. Therefore, we augment \wwj with embedding-based blocking (\S 
      \ref{subsec:bas}) and adaptive data allocation 
      (\S \ref{subsec:allocation}).
\end{enumerate}
Finally, we extend our algorithms to handle selection join queries with joins 
with recall guarantees (\S \ref{subsec:selection}).

\subsection{Weighted Wander Join} \label{subsec:wwj}

\input{tex/figures/fig_wwj_demo.tex}

As we discussed in Section \ref{subsec:bg-sampling}, existing sampling 
algorithms for join queries often fall back to inefficient uniform sampling 
due to the difficulty and prohibitive costs of obtaining exact semantic indices 
or hash functions for unstructured data. For example, Wander Join, one of the 
state-of-the-art join algorithm for structured data, formulates join operations 
into random walks over the records of different join tables 
\cite{structured-join-index}. For each iteration of the Wander Join, we randomly 
walk to a record that satisfies the join conditions (Figure 
\ref{fig:wwj}a). Knowing which records can satisfy the 
join condition requires indices over the join columns. Without such indices, 
Wander Join would simply randomly choose any record in each iteration, which is 
equivalent to uniform sampling (Figure \ref{fig:wwj}b).

While exact semantic indices is expensive to compute for unstructured data,
we can use embeddings to construct approximate indices and extend Wander Join 
accordingly. Specifically, in each iteration, instead of randomly
sampling a record satisfying join conditions, we sample a record with a 
probability proportional to the embedding similarity (Figure 
\ref{fig:wwj}c). We refer to this algorithm as the 
\textit{Weighted Wander Join} (\wwj) and illustrate the procedure for a \asum 
aggregate in Algorithm \ref{alg:wj}. 
\input{tex/figures/alg_wwj.tex}

\minihead{Connection to Importance Sampling} 
We find that \wwj is an instantiation of importance sampling \cite{kloek1978bayesian}. To understand why, we consider a 
general case of $k$ join tables: $T_1, \ldots, T_k$. \wwj is equivalent to 
importance sampling with the sampling space being the set of all tuples in the 
cross product of $T_1, \ldots, T_k$. Each tuple in the sampling space 
corresponds to a complete walk from $T_1$ to $T_k$. Given a join order, the 
sampling weight of each tuple is the multiplication of embedding similarity 
of pairs:
\begin{equation*}
    W(t_1, \ldots, t_k) = \frac{1}{|T_1|}\cdot \prod_{i=2}^{k}\mathrm{sim}\left(E\left(t_{i-1}\right), 
    E\left(t_i\right)\right)
\end{equation*}
where we normalize the weight by the size of $T_1$ since $T_1$ is 
sampled uniformly. Based on the statistical equivalence between \wwj and 
importance sampling, we can obtain unbiased estimates and valid CIs in a 
standard way \cite{kloek1978bayesian}, as shown in Algorithm \ref{alg:wj}, 
Lines 9-10.

While \wwj requires computing embedding similarities at each step without
pre-built index exists, it is more efficient than importance sampling. Naive 
sampling requires enumerating the cross-product, resulting in exponential 
complexity $\mathcal{O}(\prod_{i=2}^k |T_i|)$. In contrast, \wwj decomposes the 
problem into sequential steps. In each iteration, we only compute similarities 
between the current tuple and the $N$ tuples of the next table. Consequently, 
for a budget $b$, the complexity scales as $\mathcal{O}(b \cdot \sum_{i=1}^k |T_i|)$. 
This cost is linear (or quadratic depending on $b$) with respect to table sizes, 
but crucially avoids the exponential explosion of the cross-product.

\minihead{Limitations} The statistical efficiency of \wwj depends on the 
precision of the embedding similarity. When the tuples with high similarity 
scores all satisfy the join conditions, \wwj converges to Wander Join with exact 
indices, achieving the best efficiency. However, when there are many false 
positives (\ie, tuples with high similarity scores violets join conditions), 
\wwj can repetitively sample negative pairs, leading to an efficiency comparable 
or even worse than uniform sampling. We empirically show such limitation of \wwj 
in Section \ref{sec:eval}.

We find that false positives are often unavoidable, particularly when the 
pre-trained embedding fails to accurately capture the semantics of the join 
condition. For example, a pre-trained embedding for text data can capture the
semantic characteristic. However, two records that appear semantically similar 
might not satisfy the join condition that requires two companies to build the 
same product. In Section \ref{subsec:sensitivity}, we show that such 
issues persist even in state-of-the-art embeddings 
\cite{nv-embed,stella-embedding,qwen-embedding,bge_embedding,SFR-embedding-2}.

\subsection{Blocking-augmented Sampling} \label{subsec:bas}
\label{subsec:overview}
To improve the efficiency of \wwj when there are many false positives, we propose
\textit{Blocking-augmented Sampling} (\bas), which augments \wwj with the 
blocking algorithm. Specifically, we separate the sampling space into two regimes: 
\begin{enumerate}[leftmargin=*]
    \item The \textit{sampling regime} ($D_s$), characterized by many false negatives,
    where we apply \wwj.
    \item The \textit{blocking regime} ($D_b$), characterized by many false positives,
    where we directly execute the Oracle.
\end{enumerate}

\minihead{Intuition} 
\bas mitigates the limitation of both \wwj and blocking. First, \bas directly 
applies the Oracle on the blocking regime to mitigate the high sampling variance 
due to repetitively sampling false positives. Second, \bas can provide 
statistical guarantees since \wwj does not ignore false negatives. To execute 
\wwj only on the sampling regime, we can simply set the sampling probability of 
tuples in the blocking regime to zero. Finally, we combine the results of both 
regimes to obtain a final estimate of the aggregate. 

\minihead{Formal Description} We describe the estimation procedure formally.
Suppose we have an Oracle budget $b$ and a separation of the entire sampling 
space $D=D_s \cup D_b$. Then, we execute \wwj on $D_b$ with a sample size of 
$b-|D_b|$ and execute the Oracle on $D_s$. Finally, we can estimate 
the aggregate over $D$ by combining the estimate of the aggregate over 
$D_s$ and the true value of the aggregate over $D_b$:
\begin{align}
    & \widehat\acount = \widehat\acount_s + \acount_b,\quad \widehat\asum = \widehat\asum_s + \asum_b \label{eq:sum-count}\\
    & \widehat\aavg = \left(\widehat\asum_s + \asum_b\right)\Big/\left(\widehat\acount_s + \acount_b\right) \label{eq:avg}
\end{align}
We observe that the combined estimators for \acount and \asum are unbiased since
the estimators over $D_s$ is linear. However, the combined 
estimator for \aavg is a ratio estimator, which is 
asymptotically unbiased at the order of $\mathcal{O}(1/b)$ as proved in the 
Section 6.8 of \cite{sampling-tech}. To reduce the bias to the order of 
$\mathcal{O}(1/b^2)$, we apply the bias correction based on the Taylor 
expansion \cite{sampling-tech} to our \aavg estimator:
\begin{equation}
    \widehat\aavg_2 = \widehat\aavg \cdot \left(1-\frac{1}{b-|D_b|}\frac{Var\left[\widehat\acount\right]}{\widehat\acount^2}\right)
    \label{eq:avg-correct}
\end{equation}

\minihead{Challenges}
To realize the advantages of \bas over \wwj and blocking, we need 
to tackle two challenges. First, we need to find an optimal separation such that
the estimation error is minimized given an Oracle budget. Second, we need to 
obtain valid CIs for our combined estimators, achieving the specified 
coverage probability $p$. We introduce our adaptive allocation algorithms to 
address them.

\subsection{\bas with Adaptive Allocation} \label{subsec:allocation}

Given an Oracle budget $b$, an optimal allocation should minimize the mean 
squared error (MSE) of the aggregate estimate. In this allocation optimization
problem, we need to decide whether a tuple in the sampling space $D$ should be
allocated to the sampling regime or the blocking regime. Unfortunately, 
tuple-level allocation can be prohibitively expensive due to the size of the 
sampling space. To address it, we formulate the optimization problem into a 
stratum-level allocation problem, where we first stratify $D$ into a set 
of strata based on the similarity scores of tuples. In this case, our allocation 
algorithm determines which strata should be allocated to the blocking regime to 
minimize the overall MSE in four steps: stratification, adaptive allocation, 
sampling or blocking execution, and resampling. We explain the entire procedure 
using a \asum aggregate.

\minihead{Stratification}
Initially, we divide $D$ into a \emph{maximum blocking regime} and a 
\emph{minimum sampling regime} ($D_0$). The maximum blocking regime contains 
tuples with the top $b_2=\alpha \cdot b$ similarity scores, where $\alpha$ is a 
parameter between 0 and 1 to control the size of the maximum blocking regime. 
Based on $\alpha$, we reserve $b_2$ for potential blocking and the 
remaining budget $b_1=b-b_2$ for adaptive allocation. Next, we stratify 
the maximum blocking regime into $K$ strata ($D_1, \ldots, D_K$) with equal 
sizes (Alg.~\ref{alg:alloc}, Lines 1-5). Following prior work about stratified 
sampling for approximate query processing \cite{abae}, \bas automatically 
determines the number of strata $K$ to ensure that each stratum has an Oracle 
budget of at least 1,000.

With the stratification, we formulate the allocation optimization problem. Given 
the stratified sampling space $D_0, \ldots, D_K$, an aggregate \texttt{AGG}, 
similarity scores $W$, and an Oracle budget $b$, our target is to find an 
optimal allocation $\beta$ that minimizes the MSE of the estimated \texttt{AGG}. 
The allocation $\beta$ specifies a subset of the maximum blocking regime 
($D_1\ldots, D_K$) that will directly execute the Oracle while the rest of 
strata will execute \wwj. We formulate this optimization problem as follows:
\begin{equation}
    \beta^* = \mathop{\arg\min}_{\beta \subset \{1, \ldots, K\}} 
    MSE_{\texttt{AGG}}(D, \beta, W, b) \label{eq:opt}
\end{equation}

\input{tex/figures/alg_allocation.tex}
\minihead{Adaptive Allocation via Pilot Sampling}
To address the optimization problem, we first calculate the MSE given an 
allocation $\beta$. We observe that our estimator of the \asum aggregate
(Equation \ref{eq:sum-count}) is unbiased. Therefore, we can 
calculate overall MSE as the summation of sampling variance of each 
stratum in the sampling regime:
\begin{equation}
    MSE_{\asum} = \sum_{\substack{0 \le i\le k, i\notin \beta}} Var\left[\frac{|D_i|}{n_i}\sum_{s \in S_i}\frac{O(s)g(s)}{W(s)}\right]
    \label{eq:sum-var}
\end{equation}
where $S_i$ is the sampled tuple of $D_i$ following the procedure of \wwj and 
$n_i$ is the (hypothetical) Oracle budget assigned to stratum $D_i$. 
For a stratum in the blocking regime, $n_i = |D_i|$. For a stratum in the 
sampling regime, we follow the scheme of importance sampling and assign the 
budget proportional to the sum of similarity scores:
\begin{equation*}
    n_i =  \frac{(b - \sum_{j \in \beta} |D_i|) \cdot \sum_{s \in D_i} W(s)}{\sum_{1\le j\le K, j\notin \beta}\sum_{s \in D_j} W(s)}
\end{equation*}

Unfortunately, the sampling variance in Equation \eqref{eq:sum-var} is unknown 
unless we apply the Oracle to the entire sampling regime. To address this, 
we apply pilot sampling, a sampling procedure that is used to obtain statistics 
to guide the main sampling later. In our algorithm, we use pilot sampling to 
estimate the sampling variance with an Oracle budget of $b_1$ (Algorithm 
\ref{alg:alloc}, Line 8). Then, we can estimate the optimal allocation using the 
estimated variance (Algorithm \ref{alg:alloc}, Line 11). In Section 
\ref{sec:theory}, we show that our estimated optimal allocation converges to 
the optimal allocation at the rate of $\mathcal{O}(1/\sqrt{b_1})$. To
avoid applying Oracle on the same data tuples twice, we cache the Oracle 
results obtained in the pilot sampling stage for next stages.

\minihead{Sampling+Blocking}
Given the optimal allocation $\beta^*$, we use the remaining Oracle budget $b_2$ 
to execute \wwj on strata that are allocated to the sampling regime and execute 
the Oracle on the strata that are allocated to the blocking regime. We merge the 
results of using the budget $b_1$ (for pilot sampling) and the budget $b_2$ to 
estimate the aggregate (Alg.~\ref{alg:alloc}, Lines 12-18.)

\minihead{CI via Resampling} 
As we merge the sample from two dependent stages, the resulting sample is 
not independently and identically distributed (i.i.d.). This means we 
cannot calculate CIs using the Central Limit Theorem (CLT) that assumes 
i.i.d. data. Applying standard CLT on non-i.i.d. data naively can lead to CIs 
without valid coverage; for example, a 95\% CI might cover the true value 
with a probability lower than 95\%. To guarantee valid coverages, we use 
resampling for CIs. 

We apply the bootstrap-t resampling scheme to calculate CIs with 
coverage guarantees \cite{boostrap-comparison}. 
Given an existing sample $X$, a confidence $p$, bootstrap-t resampling 
works as follows:
\begin{enumerate}[leftmargin=*]
    \item Calculate the estimated aggregate $\hat\mu$ and its variance 
    $\hat\sigma^2$.
    \item Sample $|X|$ data from $X$ with replacement, resulting in a resample $R_i$
    \item Calculate the estimation $\hat\mu_i$ and the variance 
    $\hat\sigma_i^2$ using $R_i$.
    \item Calculate the t-statistic: $t_i=\frac{\hat\mu_i - \hat\mu}{\hat\sigma_i}$
    \item Repeat (2)-(4) for a sufficient number of times. Following prior work
    \cite{abae}, we use 1,000 resamples.
    \item Calculate the $\frac{1-p}2$ and $1-\frac{1-p}2$ percentile of t-statistics,
    resulting in $t_l$ and $t_r$, respectively.
    \item Return the CI as: $[\hat\mu-t_l\hat\sigma, \hat\mu-t_r\hat\sigma]$
\end{enumerate}
Intuitively, bootstrap-t resampling estimates the CI by measuring the
empirical CI when sampling from the empirical distribution of observed data. It 
further uses t-statistics to adjust the skewness and thus achieves the 
coverage guarantee at the order of $\mathcal{O}(1/b^2)$ \cite{boostrap-comparison}. 

We justify this validity by leveraging Theorem 23.9 of \cite{asymptotic}, 
which states that bootstrap CIs are valid if the statistical functional is 
\textit{Hadamard differentiable} with respect to the CDF. Since our estimators 
for \asum and \acount are linear combinations (and \aavg is a ratio of 
differentiable functionals), they satisfy this condition. We defer the full 
proof to Appendix \ref{sec:app-bootstrap}.

As shown in Algorithm \ref{alg:alloc}, our resampling step for CI 
computation operates independent of the adaptive allocation. The allocation 
objective is to minimize MSE, which inherently optimizes accuracy regardless of 
whether a CI is requested. Therefore, while skipping CI computation eliminates 
the resampling overhead, it does not directly improve estimation accuracy. For 
applications constrained by total latency rather than 
monetary cost, the computational time saved by omitting resampling could be 
leveraged to increase the Oracle budget, thereby indirectly improving estimation 
accuracy.


\minihead{Handling \aavg} 
We estimate \aavg with a ratio estimator (Eq.~\ref{eq:avg} and 
\ref{eq:avg-correct}), which is asymptotically unbiased. Given the independence 
between \asum and \acount, we estimate the MSE of \aavg as follows:
\begin{equation*}
    MSE_{\aavg} = \frac{1}{n}\left(\frac{\widehat{\asum}}{\widehat{\acount}}\right)^2
    \left(\frac{MSE_{\acount}}{\widehat{\acount}^2} + 
    \frac{MSE_{\asum}}{\widehat{\asum}^2}\right)
\end{equation*}
where $n=b - \sum_{i\in \beta} |D_i|$ is the Oracle budget for the sampling 
regime. This MSE estimator is based on widely used variance estimator for \aavg 
\cite{ripplejoin}, which is accurate to the order of $\mathcal{O}(1/b^2)$ 
\cite{sampling-tech}. To handle \aavg, we just replace estimators in Algorithm 
\ref{alg:alloc}.

\minihead{Handling \amin and \amax}
\bas leverages the correlation between embedding similarity and the target 
attribute. As an example, we estimate \amax by 
$\widehat\amax = \max_{t \in D_b \cup S} val(t)$, the maximum observed. Unlike 
linear aggregates where we 
minimize global variance, the allocation goal for \amax is to maximize the
probability that the true maximum is contained within the blocking regime $D_b$. 
Utilizing the pilot sample, we model the distribution of values in each stratum 
using Extreme Value Theory \cite{de2006extreme} and allocate strata to $D_b$ 
based on their probability of containing values exceeding the current sample 
maximum. Since the distribution of the sample maximum is non-normal, bootstrap 
CIs are invalid. Instead, we construct CIs using the observed maximum as a 
lower bound and the global maximum of the dataset (ignoring join 
conditions) as an upper bound.

\minihead{Handling \amedian}
Estimating the median requires constructing the CDF of the target attribute. 
\bas estimates the CDF by combining the exact distribution from blocking 
with the weighted empirical distribution from sampling.
Let $F(t)$ be the CDF. We estimate it as:
\begin{equation*}
\footnotesize
\widehat{F}(t) = \frac{1}{\widehat{N}} \left( \sum_{x \in D_b} \mathbb{I}(x \le t) + \sum_{y \in D_s} \frac{1}{\pi_y} \mathbb{I}(y \le t) \right)
\end{equation*}
where $W_y$ is the sampling probability. The estimated \amedian is then $\widehat{F}^{-1}(0.5)$. The 
variance of the median estimator depends on the variance of the CDF estimate at 
the median value. Therefore, our allocation algorithm focuses on minimizing the 
variance of the CDF in the strata that overlap with the estimated median range 
derived from the pilot sample. As the median is a statistical functional that 
is Hadamard differentiable with respect to the CDF, we apply the same 
bootstrap-t resampling strategy to derive valid CIs.

\minihead{Handling \agroupby}
Group-by queries often suffer from high variance in ``tail'' groups (groups with 
low support) under standard sampling, \bas exploits the semantic clustering of 
embeddings to target specific groups effectively. We 
first estimate the aggregate for each group $g$. Unlike global aggregation where 
we minimize total variance, our adaptive allocation for group-by queries 
minimizes the \textit{maximum relative error} across all discovered groups. By 
using pilot samples to estimate the conditional probability of groups within 
strata, \bas prioritizes blocking in strata containing high densities of 
``hard-to-estimate'' groups. Finally, we use bootstrap-t to provide simultaneous 
confidence intervals, accounting for the multiple hypothesis testing in 
multi-group estimation.

\minihead{Algorithm Complexity}
Consider a chain join of
$k$ tables with size $N_1, \ldots, N_k$ and $b \ll \prod_{i=1}^k N_i$. \bas has 
a time complexity of $\mathcal{O}\left(\prod_{i=1}^k N_i\right)$ and a space 
complexity of $\mathcal{O}\left(\sum_{i=1}^{k-1} N_iN_{n+1} + b\log b\right)$. This complexity is higher than that of 
uniform sampling, which has both time and space complexities of $\mathcal{O}(b)$, 
because \bas incorporates similarity scores. Additionally, BAS has higher 
complexity than \wwj, which has time and space complexities 
$\mathcal{O}\left(b \cdot \sum_{i=2}^k N_i\right)$, primarily due to the 
stratification. However, when compared to threshold-based blocking, which shares 
the same time complexity with \bas and has a space complexity of $\mathcal{O}(b)$, 
\bas only incurs a small overhead to sort the top $\alpha b$ records. 
Furthermore, \bas is more efficient than \abae, which involves sorting all 
similarity scores \cite{abae}. Finally, \bas shares the same time complexity with 
\blazeit, which applies control variates over the cross product \cite{blazeit}.

Although \bas has a time complexity that increases exponentially with the number 
of tables, the leading order of the time complexity is attributed to 
comparing floating-point numbers using CPUs, which is significantly faster than 
the Oracle. Our latency profiling reveals that \bas's CPU computations 
take up to 90 seconds on our datasets, costing \$0.03 on AWS—making it 
$10^6\times$ cheaper than running Oracle with GPT-4o mini.

When the table size or the number of tables is so large that the cross product
cannot fit into memory, threshold-based blocking in \bas becomes prohibitively 
expensive. To mitigate this, we apply the nearest neighbor-based blocking 
\cite{christen2011survey}. It effectively joins each record of the 
left table with the top $b'$ records of the right table for each join, where 
$b'=\left(\alpha n / N_1\right)^{1/(k-1)}$. This method has a time complexity of 
$\mathcal{O}\left(b \cdot \sum_{i=2}^k N_i\right)$ and a space complexity of 
$\mathcal{O}\left(b\right)$, making the complexity of \bas comparable to 
\wwj and significantly lower than \abae and \blazeit. We show that \bas 
with nearest neighbor-based blocking remains more efficient than the baselines 
(\S \ref{subsec:e2e}).

\subsection{\bas for Selection Queries} \label{subsec:selection}
We first review approximate selection queries. Consider a selection 
query $Q$ that, when executed exactly, outputs a set of records $T$. If $Q$ 
is processed approximately, the result is $T'$. As established in prior work 
\cite{supg}, we say that $T'$ achieves the recall target $\gamma$ and confidence 
$p$ if $\mathbb{P}[|T\cap T'|/|T|\ge \gamma]\ge p$.

To process such queries, prior work calculates a score for each 
data and outputs all tuples with scores higher than a threshold 
$\tau_s$ \cite{supg}. The score approximately reflects the probability of a 
tuple satisfying the predicate, while $\tau_s$ is estimated using 
sampling to achieve recall guarantees.

With \bas, we can improve the selection quality by maximize the threshold 
$\tau_s$. This is achieved by adaptive allocating data to minimize the recall target 
$\gamma_s$ for the sampling regime. To formulate the allocation problem, we first
translate the overall recall target $\gamma$ to the recall target of the sampling
regime $\gamma_s$ through the following lemma. We defer the full proof to 
Appendix \ref{sec:app-select}.
\begin{lemma}
With a probability higher than $p$, we can achieve the overall recall target 
$\gamma$ if $\gamma_s$ satisfies
\begin{equation*}
    \gamma_s \ge \gamma - (1-\gamma)\frac{\acount_b}{\mathrm{UB}(\acount_s, Var[\acount_s], b, p)}
\end{equation*}
where
\begin{equation*}
    \mathrm{UB}(\mu, \sigma^2, b, p) = \mu + \frac{\sigma}{\sqrt{b}}\sqrt{2\log \frac{2}{1-p}}
\end{equation*}
\label{lemma:select}
\end{lemma}

Given a stratification 
$D_0, \ldots, D_K$ and a budget $b_2$, we find the optimal allocation $\beta^*$ 
to minimize $\gamma_s$:
\begin{equation*}
    \beta^* = \mathop{\arg\max}_{\beta \subset \{1, \ldots, K\}}
    \frac{\acount_i}{\sum_{i \notin \beta} \mathrm{UB}(\acount_i, Var[\acount_i], b_i, \frac{p+k-1}{k})}
    \label{eq:recall-opt}
\end{equation*}
where $\acount_i=\sum_{s\in D_i} O(s)$ is the number of matching tuples in 
stratum $i$ and $b_i$ is the assigned budget for stratum $i$. Similar to 
aggregation queries, we approximately solve the allocation problem with 
estimated $\acount_i$ and $Var[\acount_i]$ using the pilot sampling.

\minihead{Handling \atopk Heavy Hitters Selection}
A heavy-hitter \atopk query identifies the $K$ entities with the most appearance 
\cite{yi2009optimal,cormode2003finding,wang2018efficient}. This 
generalizes the selection problem with a dynamic threshold determined by 
the $K$-th value. We estimate \acount for every candidate value using the 
combined estimator from Equation \ref{eq:sum-count}. Let $\widehat{V}_e$ be the 
estimated value for entity $e$. We return the set of $K$ entities with the 
largest $\widehat{V}_e$. To ensure the correctness of the ranking, we minimize 
the swapping probability between the $K$-th ranked entity and the $(K+1)$-th 
ranked entity. We adapt our allocation objective to minimize the estimation 
variance specifically for candidates whose estimates are close to the 
threshold $\tau$ (the estimated $K$-th value). We provide guarantees on the 
recall: with probability $p$, the returned set contains the true \atopk 
entities. We achieve this by computing simultaneous confidence intervals for the 
candidate entities using bootstrap-t, ensuring that the lower bound of the 
$K$-th item exceeds the upper bound of the $(K+1)$-th item.

\subsection{Practical Guidelines} \label{subsec:practical}
We offer the following deployment guidance. \bas is beneficial when the cost of 
Oracle significantly outweighs the overhead of embedding similarity calculations. 
For LLM-based predicates (e.g., GPT-4), Oracle costs are typically orders of 
magnitude higher than similarity computation, making \bas highly effective. For 
cold-start scenarios, we recommend allocating 15--30\% of the budget to the 
maximum blocking regime ($b_2$) to provide a balanced search space for the 
optimizer without over-constraining the sampling budget. \bas automatically 
tunes $K$ to ensure each stratum receives $>1000$ samples. For small budgets, we 
enforce a minimum $K=5$ to ensure sufficient granularity. Finally, \bas is 
model-agnostic. For semantic joins (e.g., product descriptions), we recommend 
dense embeddings (e.g., CLIP \cite{clip}). For lexical joins, sparse vectors 
(e.g., TF-IDF) often yield sharper boundaries. To achieve better performance 
with \bas, we recommend evaluating the false positive and false negative 
rates of different embeddings on a small data sample.
\color{black}

%% file: tex/figures/fig_wwj_demo.tex
\begin{figure}[t!]
    \centering
    \includegraphics[width=\linewidth]{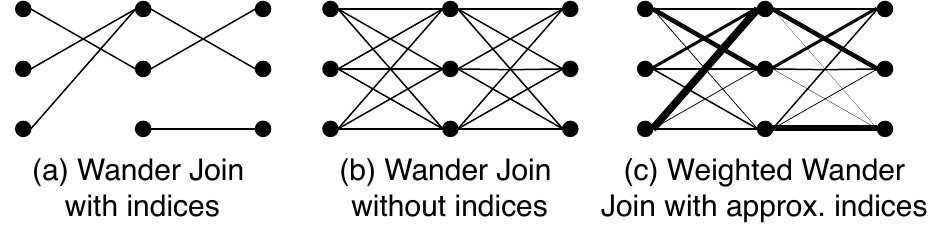}
    \caption{Illustration of Wander Join with indices, without indices, and with
    approximate indices. Black dots are data records while solid connections are
    viable paths during random walks. The widths of edges in (c) show the 
    probability of choosing paths.}
    \label{fig:wwj}
\end{figure}

%% file: tex/figures/alg_wwj.tex
\begin{algorithm}[t]
    \DontPrintSemicolon
    \SetKwInOut{Input}{Input}\SetKwInOut{Output}{Output}
    \SetKwComment{SideComment}{\textcolor{blue}{//~}}{}
    \SetKwComment{CommentStyle}{\textit{// }}{}
    \SetKwComment{LineComment}{\textcolor{blue}{\textit{/* }}}{\textcolor{blue}{  */}}
    \SetFuncSty{text}
    \SetKwFunction{FCI}{StandardCI}
    \SetKwFunction{FWSample}{WeightedSample}
    \SetKwFunction{FSample}{Sample}

    \Input{input tables $T_1, \ldots, T_k$, 
    embedding $E(\cdot)$, attribute function $g(\cdot)$, 
    embedding similarity $\mathrm{sim}(\cdot, \cdot)$,
    oracle budget $b$, probability $p$}

    $S \leftarrow \emptyset$\;
    \For{$i$ in $1, \ldots, b$}{
        $t_{i,1} \leftarrow \FSample(T_1, 1)$\;
        \For{$j$ in $2, \ldots, k$}{
            $W_{i,j} \leftarrow \{\mathrm{sim}\left(E(t_{i,j-1}), E(t)\right):t\in T_j\}$\;
            $t_{i,j} \leftarrow \FWSample(T_j, 1, W_{i,j})$\;
        }
        $S \leftarrow S \cup \{ (t_{i,1}, \ldots, t_{i,k}) \}$\;
        $W(s) \leftarrow \frac{1}{|T_1|} \prod_{j=2}^k \mathrm{sim}\left(E(t_{i,j-1}), E(t_{i,j})\right)$
    }
    $X \leftarrow \left\{O(s) \cdot g(s) \cdot W(s)^{-1}: s \in S\right\}$\;
    \Return{\FCI{$X, p, \prod_{i=1}^k |T_i|$}}

    \caption{Weighted Wander Join algorithm (\asum).}
    \label{alg:wj}
\end{algorithm}

%% file: tex/figures/alg_allocation.tex
\begin{algorithm}[t]
    \DontPrintSemicolon
    \SetKwInOut{Input}{Input}\SetKwInOut{Output}{Output}
    \SetFuncSty{text}
    \SetKwFunction{FBudget}{BudgetAssign}
    \SetKwFunction{FCP}{CrossProduct}
    \SetKwFunction{FTop}{Top}
    \SetKwFunction{FSort}{SortDesc}
    \SetKwFunction{FWSample}{WeightedSample}
    \SetKwFunction{FSample}{Sample}
    \SetKwFunction{FResample}{ResamplingCI}
    \SetKwProg{Fn}{Function}{:}{}
    \Input{input tables $T_1, \ldots, T_k$, 
    similarity score $W(\cdot)$, attribute $g(\cdot)$, 
    number of stratum $K$, maximum blocking ratio $\alpha$, 
    oracle budgets $b_1$ and $b_2$, probability $p$}
    
    \tcc{\textcolor{gray}{Stratification}}
    $D' \leftarrow \FTop(\FCP(T_1, \ldots, T_k), W, \alpha(b_1+b_2))$\;
    $D' \leftarrow \FSort(D', W)$\;
    $D_0 \leftarrow D \backslash D'$\;
    \For{$i$ in $1, \ldots, K$}{ 
            $D_i \leftarrow D'\left[\alpha(i-1)(b_1+b_2)/K : \alpha \cdot i(b_1+b_2)/K  \right]$
        }
    \tcc{\textcolor{gray}{Adaptive Allocation via Pilot Sampling}}
    \For{$i$ in $0, \ldots, K$}{
    $n^{(1)}_i \leftarrow \frac{\sum_{x \in D_i} W(x)}{\sum_{x \in D} W(x)} b_1$\;
    $S^{(1)}_i \leftarrow \FWSample(D_i, n_i, W)$\;
    $X_i^{(1)} \leftarrow \left\{g(s) \cdot O(s) \cdot W(s)^{-1}: s\in S^{(1)}_i\right\}$\;
    $\hat\sigma_i^2 \leftarrow |D_i|^2 \cdot Var\left[X_i^{(1)}\right]$\;
    }
    $\beta^* \leftarrow \mathop{\arg\min}_{\beta}\ 
    \sum_{\substack{i\notin \beta}} \hat\sigma_i^2 \bigl/ \FBudget(b_2, W, \beta, i, D)$
    
    \tcc{\textcolor{gray}{Sampling + Blocking}}
    \For{$i\notin \beta^*$}{
        $n^{(2)}_i \leftarrow \FBudget(b_2, W, \beta^*, i, D)$\;
        $S^{(2)}_i \leftarrow \FWSample(D_i, n_i, W)$\;
        $X_i^{(2)} \leftarrow \left\{g(s) \cdot O(s) \cdot W(s)^{-1}: s\in S^{(2)}_i\right\}$\;
    }
    \For{$i\in \beta^*$}{
        $S^{(2)}_i \leftarrow \FSample(D_i, |D_i|)$\;
        $X_i^{(2)} \leftarrow \left\{g(s) \cdot O(s): s\in S_i^{(2)}\right\}$\;
    }
    \tcc{\textcolor{gray}{CI via Resampling}}
    $X \leftarrow X^{(1)} \cup X^{(2)}$\;
    \Return{$\FResample(X, p, D)$}

    \Fn{\FBudget{$b, W, \beta, i, D$}}{
        \Return{$\left(b-\sum_{j\in \beta} |D_j|\right) \frac{\sum_{s \in D_i} W(s)}{\sum_{j \notin \beta}\sum_{s\in D_j} W(s)})$}\;
    }
    
    \caption{\bas with Adaptive Allocation for a \asum Aggregate}
    \label{alg:alloc}
\end{algorithm}

%% file: tex/theories.tex
\begin{figure*}
    \centering
    \includegraphics[width=\linewidth]{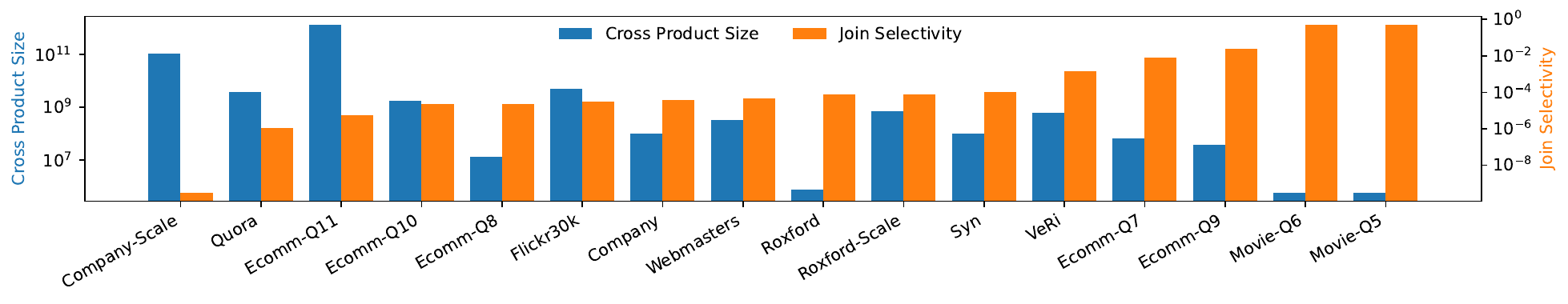}
    \caption{Cross product size and join selectivity of the 16 
    evaluated datasets: real-world, SemBench, and synthetic workloads.}
    \label{fig:datasets}
\end{figure*}

\section{Theoretical Analysis of \bas} \label{sec:theory}
In our analysis, we first show MSE of \bas converges to the MSE 
with optimal allocation asymptotically at the rate 
$\mathcal{O}\left(b_1^{-1/2}\right)$. Furthermore, we compare the MSE of \bas 
with that of \wwj, showing that \bas outperforms or matches \wwj asymptotically. 

\subsection{\bas Converges to the Optimal Allocation} \label{subsec:optimal-alloc}

\bas applies pilot sampling to approximately solve the minimization problem in 
Equation \ref{eq:opt}. Theorem \ref{theom:opt} shows that the MSE with estimated 
minimizer $\widehat{\beta}^*$ converges to the MSE with the true minimizer 
$\beta^*$ at the rate of $\mathcal{O}\left(1/\sqrt{b_1}\right)$. 

\begin{theorem} \label{theom:opt}
The MSE with the estimated minimizer $\hat\beta^*$ converges to that with the true minimizer 
$\beta^*$ at the rate:
\begin{equation}
    \frac{MSE(D, \widehat\beta^*, W, b_2) - 
    MSE(D, \beta^*, W, b_2)}{MSE(D, \beta^*, W, b_2)} 
    = \mathcal{O}\left(\frac{1}{\sqrt{b_1}}\right) \notag
\end{equation}
\end{theorem}
Intuitively, we observe that statistics for allocation determination, 
specifically the sample mean and sample variance, converge to their true values 
at a rate proportional to $1/\sqrt{b_1}$. Theoretically, we demonstrate the MSE 
of \bas converges to that of the optimal allocation at the same rate. This is 
because the arithmetic operations involved in calculating the MSE (\ie, 
summation, multiplication by constants, and division by constants) do not change 
the convergence rate. We defer a full proof to Appendix \ref{sec:app-optimal}.

\subsection{\bas Outperforms or Matches \wwj} \label{subsec:comp-wwj}

\minihead{Notation} In addition to the notations we introduced in Section 
\ref{sec:alg}, We define that $\tilde{D}$ is the set of all positive tuples, the 
subscript $s$ means the statistic is about the sampling regime of \bas.

We present the comparison of the MSE of \bas and \wwj for the \asum aggregate in 
Theorem \ref{theo:comp} and defer a full proof to Appendix \ref{sec:app-comparison}. 
We can derive the same result for the \acount or \aavg similarly.

\begin{theorem} \label{theo:comp}
    If there exists an allocation $\beta$ such that the following two conditions hold
    \begin{align}
        \mathbb{E}_{s\in \tilde{D}_s}\left[\frac{1/|\tilde{D}_s|}{W_s(s)}\right] 
        &\le \mathbb{E}_{s\in D}\left[\frac{1/|D|}{W(s)}\right]
        \label{eq:cond1-t}\\
        \frac{|\tilde{D}_s|^2}{b_s} &\le \frac{|D|^2}{b} \label{eq:cond2-t}
    \end{align}
    \bas outperforms \wwj asymptotically, \ie,
    \begin{equation}
        MSE^{(\bas)}_{\texttt{SUM}} = C \cdot MSE^{(\wwj)}_{\texttt{SUM}} + \mathcal{O}\left(b^{-1}b_1^{-1/2}\right) \notag
    \end{equation}
    where $W_s$ is the similarity scores normalized for the sampling regime and
     $C$ is a coefficient less than 1:
    \begin{equation}
        C < \frac{|\tilde{D}_s|^2 \big/ b_s}{|\tilde{D}|^2 / b} 
            \frac{\mathbb{E}_{s\in \tilde{D}_s}
                \left[\frac{1/|\tilde{D}_s|}{W_s(s)}\right]}
                { \mathbb{E}_{s\in \tilde{D}}\left[\frac{1/|\tilde{D}|}{W(s)}\right]} 
        \le 1 \notag
    \end{equation}
    Otherwise, \bas matches \wwj asymptotically, \ie,
    \begin{equation}
        MSE^{(\bas)}_{\texttt{SUM}} \le MSE^{(\wwj)}_{\texttt{SUM}} + \mathcal{O}\left(b^{-1}b_1^{-1/2}\right) \notag
    \end{equation}
\end{theorem}

Theorem \ref{alg:alloc} suggests two cases when we compare the MSE of \bas and 
that of \wwj. We discuss both cases in detail.

\minihead{Case 1: \bas outperforms \wwj asymptotically} 
Theorem \ref{theo:comp} provides sufficient conditions when \name achieves 
asymptotically better MSE than \wwj. We observe that both conditions can match 
characteristics of the embedding similarity. First, Condition \eqref{eq:cond1-t} 
compares the similarity between the proposed sampling distribution $W$
and the ideal distribution (\ie, $1/|\tilde{D}|$) for the matching tuples in the 
sampling regime or the entire sampling space. When we reduce false positives via 
blocking, the proposed sampling distribution may become closer to the idea 
distribution, satisfying Condition \ref{eq:cond1-t}. Second, Condition 
\ref{eq:cond2-t} compares the density of matching tuples for tuples in the 
sampling regime or the entire sampling space. When the similarity scores have a 
high recall, the blocking regime includes most of the positive tuples, while the 
positive tuples in the sampling region are sparse, satisfying Condition 
\ref{eq:cond2-t}.

\minihead{Case 2: \bas matches \wwj asymptotically} 
When sufficient conditions are not fully satisfied, the difference between \bas 
and \wwj converges to 0 at the rate $\mathcal{O}(b^{-1}b_1^{-1/2})$, 
which is faster than $\mathcal{O}(b^{-1})$, the rate at which the MSE converges to 0.

%% file: tex/eval.tex
\section{Evaluation} \label{sec:eval}
In this section, we present evaluation results of our algorithms.
We first introduce our experiment setup (\S \ref{subsec:setup}). Then, we 
demonstrate that \bas achieves statistical guarantees, while \block fails 
(\S \ref{subsec:stats-analysis}). Moreover, we demonstrate that \bas outperforms 
baselines in terms of MSE (\S \ref{subsec:e2e}). Finally, we perform the 
ablation study (\S \ref{subsec:optimality}) and sensitivity study (\S \ref{subsec:sensitivity}).

\subsection{Experiment Settings} \label{subsec:setup}

\minihead{Datasets and queries}
We curated a comprehensive suite of 16 datasets originating from realistic 
integration tasks \cite{dataset-quora,dataset-quora-work,dataset-quora,
dataset-quora-work,dataset-stackoverflow,dataset-stackoverflow-work1,
dataset-company-work1,dataset-company-work2,radenovic2018revisiting,
dataset-flickr30k,veri1,veri2}, the SemBench benchmark 
\cite{lao2025sembench}, and synthetic stress-tests. To assess robustness across 
data distributions, our workload spans a wide range of join selectivities 
(from $5 \times 10^{-10}$ to $0.5$) and scales (up to 1.3 trillion pair 
cross-products). We summarize the dataset statistics in Figure \ref{fig:datasets} 
and defer detailed construction to Appendix \ref{sec:app-dataset}. For each 
dataset, we use the embedding model that achieves or matches the state-of-the-art 
performance. We evaluated all three aggregates for each dataset and present the 
results of the most natural aggregate for each dataset. We categorize our queries 
as follows:
\begin{enumerate}[leftmargin=*]
    \item \textit{Real-World Analytical Joins}: Entity resolution,
    duplicate detection, and video analysis tasks (Company \cite{dataset-company-work2}, 
    Quora \cite{dataset-quora}, Webmasters \cite{dataset-stackoverflow-work1}, 
    Roxford \cite{radenovic2018revisiting}, Flickr30K \cite{dataset-flickr30k}, 
    VeRi \cite{veri1}).
    \item \textit{SemBench}: Queries requiring joint semantic 
    understanding across tables (e.g., matching reviews by sentiment, linking 
    products by images/descriptions). We include all such join queries with a 
    result set of more than 100 records, covering E-commerce and Movie domains.
    \item \textit{Synthetic Stress-Tests}: Company-Scale (6-way join) and Roxford-Scale (10M rows) to test extreme scalability; and $Syn(p_{fp}, p_{fn})$ to systematically vary embedding quality.
\end{enumerate}

\minihead{Baselines}
We compare our algorithms, \wwj and \bas, with the following baselines:
\begin{denseitemize}
    \item \uniform: We uniformly sample tuples from the entire cross product of the tables.
    \item \block: A proxy for the state-of-the-art deep EM system, Ditto \cite{dataset-company-work2}. It samples tuples above a similarity threshold 
    calibrated on a validation set (10\% of positive tuples).
    \item \abae and \blazeit: We adapt the state-of-the-art AQP algorithm for 
    aggregation queries with ML predicates by treating the join condition as a 
    selection predicate \cite{abae,blazeit}.
    \item SUPG and LOTUS: We use the state-of-the-art AQP algorithms for 
    selection queries with ML predicates \cite{supg,patel2024lotus} as baselines 
    for the extension of our algorithms for selection queries.
\end{denseitemize}
We do not directly compare against full EM pipelines (e.g., 
Magellan \cite{usecase-em}) because they are designed as offline processes that 
materialize the full result, whereas our algorithms are for an online AQP with a 
different goal. However, we compare against the core phases of EM (blocking and 
matching) via our \block baseline, which effectively represents the performance 
of a full EM pipeline on AQP tasks. In addition, we do not consider non-ML 
text-based EM systems or additional data augmentation used in Ditto 
\cite{dataset-company-work2} because these techniques are specific to the text 
modality and do not generalize to the multimodal datasets we consider.

\subsection{Statistical Guarantees} \label{subsec:stats-analysis}

\input{tex/figures/fig_guarantees.tex}

We analyze whether the baseline method (\block) and our proposed method (\bas)
can achieve statistical guarantees. Specifically, we focus on whether these 
methods can provide valid CIs, a standard approach in statistically guaranteed
AQP \cite{supg,abae,blazeit,ripplejoin,structured-join-index}.
\newcommand\te{\mathrm{True\ value}}
\newcommand\cib{\mathrm{CI\ bound}}
\newcommand\est{\mathrm{Estimate}}

\minihead{Metric}
We use the ratio between the true error and the CI bounds as our evaluation 
metric. Given a query with a groundtruth result $\mu$, an estimated 
result $\hat\mu$, and a confidence interval $[l,u]$, the error ratio is 
calculated as: $\frac{|\hat\mu - \mu|}{|u-l|/2}$.

According to the definition of a CI, a valid 95\% CI must include the true value 
with a probability of 95\%.  Consequently, the probability that the true error 
falls within the CI bounds should be at least 95\%. Thus, for a method to 
provide valid CIs, the 95th percentile of the ratio between the true error and 
the CI bounds must not exceed 1.

\minihead{\bas Achieves Statistical Guarantees}
We present the results for \block and \bas across various Oracle budgets in 
Figure \ref{fig:guarantees}. 
For each 
Oracle budget, we repeated the experiments 500 times. As shown, \bas consistently 
achieves 95th percentile error ratios of less than 1, indicating the validity of 
provided CIs. In contrast, \block fails to achieve valid CIs for every query, 
leading to true errors that are up to 79\% higher than the CI errors.

We observe that the error ratio of \block increases as we increase the Oracle 
budget. This trend is expected because of the inherent bias of \block: it 
determines the data to execute Oracle through a similarity score threshold. 
Even when the similarity threshold is calibrated on a validation 
dataset, \block can still ignore positive pairs in the evaluation dataset, 
resulting in inherently biased estimation results. As the Oracle budget 
increases, the true error of \block decreases and converges to an unknown 
inherent bias, while the width of the CI approaches to zero. Consequently, when the 
Oracle budget is sufficiently large, true errors can be smaller than CI errors, 
causing invalid CIs. \bas resolves this problem by incorporating sampling 
algorithms to provide unbiased or asymptotically unbiased estimations, leading 
to valid CIs.
\input{tex/figures/fig_low_guarantees.tex}
\input{tex/fig_rrmse.tex}

\minihead{\bas Achieves Statistical Guarantees Under Limited Budgets}
We show that our asymptotic statistical guarantees hold under limited budgets.
In the left part of Figure \ref{fig:guarantee-budget}, we report the maximum 
error ratio across all datasets. We show that \bas maintains valid CIs even when 
the oracle budget is as low as 1,000. These results empirically validate the 
robustness of the bootstrap-t method used in \bas, ensuring reliability even in 
small-budget scenarios. In right part of Figure \ref{fig:guarantee-budget}, we 
show this validity holds regardless of the pilot sample size (varying from 0.1\% 
to 10\% of the budget). Intuitively, this is because the size of the  pilot 
sample does not affect the derivation of confidence intervals, thus not 
affecting our statistical guarantees. 

\subsection{End-to-end Performance} \label{subsec:e2e}
We analyze end-to-end performance on a diverse workload including real-world analytical joins, SemBench semantic queries, and multi-way joins. We report the average relative RMSE over 100 repetitions with a maximum blocking ratio $\alpha=20\%$.

\minihead{\bas Improves Overall Estimation Errors} As shown in Figure 
\ref{fig:rrmse}, \bas reduces estimation errors by 1.04--19.5$\times$ compared 
to the best performing baseline among \uniform, \abae, and \blazeit. Compared to 
\wwj, \bas achieves similar or lower estimation errors by up to 17.2$\times$. 
Except for \dquora, \bas reduces errors by up to 2.2$\times$ compared to Oracle 
\block. This improvement is due to the relatively low Oracle threshold to achieve 
valid confidence intervals. While Oracle \block can be more efficient if positive 
and negative tuples are well distinguished by similarity scores (as in the case 
of \dquora), determining the Oracle threshold requires exhaustive execution of 
the Oracle, making it prohibitively expensive.

\minihead{\bas Reduces Estimation Errors Across Selectivities} 
We categorize results by join selectivity to demonstrate \bas's adaptability:
\begin{enumerate}[leftmargin=*]
    \item \textit{Low-Selectivity} ($< 0.001$): On datasets with rare matches 
    (Fig.~\ref{fig:rrmse-company}--\ref{fig:rrmse-flickr30k}, 
    \ref{fig:rrmse-eq8}, and \ref{fig:rrmse-eq10a}--\ref{fig:rrmse-roxfordl}), 
    \bas significantly outperforms all baselines by up to 19.5$\times$. For the 
    6-way Company-Scale join (selectivity $5\times 10^{-10}$), Uniform sampling failed to retrieve sufficient positive samples to form an estimate, while \bas maintained low error.
    \item \textit{High-Selectivity} ($0.001-0.5$): In high-density 
    scenarios (Fig.~\ref{fig:rrmse-veri}--\ref{fig:rrmse-mq6}. and 
    \ref{fig:rrmse-eq9}) where blocking offers little theoretical advantage, 
    \bas demonstrates robustness. It adaptively allocates budget towards sampling, 
    preventing the performance regression observed in \block. BaS achieves errors 1.05–2.5$\times$ lower than baselines and matches \wwj.
\end{enumerate}

\minihead{\bas Improves Real-World Multi-Way Joins}
We evaluated scalability using real-world 3-way (Ecomm-Q10), 4-way (Ecomm-Q11),
and synthetic 6-way (Company-Scale) joins. As shown in 
Figure \ref{fig:rrmse-eq10}--\ref{fig:rrmse-companyL}, \bas reduces error by 
1.4–2.1$\times$ compared to \uniform and \wwj, while \abae and \blazeit do not 
natively support multi-way joins. This performance gap highlights the critical 
role of blocking in large search spaces. As the number of joined tables 
increases, the density of valid join paths decreases, causing inefficient sampling.

\minihead{\bas Improves \amax, \amin, \amedian, and \agroupby Queries}
To demonstrate general applicability of \bas, we evaluated non-linear aggregators 
and \agroupby queries. We show that \bas reduces estimation error by up to 
1.4$\times$ for \amedian (Fig.~\ref{fig:rrmse-eq9}), 2.0$\times$ for \amax 
(Fig.~\ref{fig:rrmse-mq6}), and 3.0$\times$ for \amin (Fig.~\ref{fig:rrmse-eq8}). 
For \agroupby queries (Fig.~\ref{fig:rrmse-eq10a}), \bas reduces the 
mean estimation error across groups by up to 2.1$\times$ compared to baselines, 
demonstrating its utility for OLAP-style analytics.

\minihead{\bas Improves Selection Queries} 
We compared \bas against LOTUS \cite{patel2024lotus} on selection queries 
(Fig.~\ref{fig:recall}). On a selection query based on the Company dataset, 
while both methods satisfy the recall target with 95\% confidence, \bas achieves 
significantly higher precision, improving by 2.6–43.8\% over LOTUS and 3.0–68.6\% 
over Uniform. In addition, we show that \wwj achieves performance similar to
SUPG (dotted purple line), empirically validating that \wwj shares the same 
statistical characteristics as SUPG for selection queries. Furthermore, for a \atopk 
heavy-hitter query based on Ecomm-Q11 with recall guarantees (95\% confidence), 
where LOTUS lacks native support, \bas improves precision by 9.6--10.7\% over baselines. 
This indicates that \bas's variance-optimal allocation is effective not only for 
aggregation, but for high-precision retrieval.

\subsection{Application of Approximate \acount: Join Order Optimization} \label{sec:eval-app}
We applied \bas to the problem of join order optimization for exactly executing a 
multi-way semantic join query (Ecomm-Q11 of SemBench), where cardinality 
estimation errors can lead to execution costs 80$\times$ higher than optimal. By 
integrating \bas with the DPccp enumeration algorithm 
\cite{moerkotte2006analysis}, we generate more accurate cardinality estimates 
for sub-joins. As shown in Figure \ref{fig:join_opt}, the improved \acount 
estimation from \bas allows the optimizer to identify efficient execution plans, 
reducing total execution cost by up to 12.7\% compared to plans generated using 
\uniform or \wwj estimates.

\subsection{Ablation Study} \label{subsec:optimality}
We empirically validate the theoretical result in Section \ref{sec:theory} that
\bas converges to the optimal allocation. To achieve this, we compare the adaptive 
allocation with the optimal allocation and the worst allocation. To approximate 
the optimal and worst allocations, we executed \bas with fixed blocking ratios 
ranging from 10\% to 50\%. For example, a fixed blocking ratio of 10\% means we 
directly execute Oracle on pairs with top 10\% similarity scores while applying
\wwj on the rest. For each ratio, we repeated experiments 100 
times.

\input{tex/figures/fig_extension.tex}
\input{tex/figures/fig_join_opt.tex}
\input{tex/figures/fig_ablation.tex}

In Figure \ref{fig:optimal}, we show the improvements of \bas using optimal, 
adaptive, and worst allocations over \wwj. We calculate the improvement as the 
relative reduction of RMSE.
As shown, \bas with adaptive allocations achieves improvements that are only 
1.61-18.9\% (7.04\% on average) less compared to \bas with optimal allocations. 
In contrast, \bas with the worst allocations underperforms (\ie, negative 
improvements) by up to 98.4\% compared to \wwj and up to 99.9\% compared to \bas
with optimal allocations.\footnote[4]{We remove the majority of the negative 
improvements of the worst allocations in the figure for simplicity.} 

\subsection{Sensitivity Analysis} \label{subsec:sensitivity}

We analyze the sensitivity of \bas in terms of the quality of embedding models,
strategies to use the similarity scores as sampling weights, and the maximum 
blocking ratio $\alpha$.

\minihead{Impact of False Negatives and False Positives}
We examined the impact of embedding quality using synthetic datasets with 
controlled False Negative Rates (FNR) and False Positive Rates (FPR). As shown
in Figure \ref{fig:synthetic}, when FNR and FPR reach 50\%, BaS outperforms 
baselines by more than 20\%. Specifically, BaS dominates Blocking when FNR is high 
(by correcting for missed matches via sampling) and dominates WWJ when FPR is 
high (by avoiding the overweighting of false positives).

\minihead{Impact of the Embedding Quality} 
We executed algorithms on the Company dataset with 9 more embeddings, including 
traditional textual similarity \cite{cosine,tfidf}, encoder-only embedding 
models \cite{crossencoder-2}, and large-scale embedding models 
\cite{nv-embed,bge_embedding,stella-embedding,SFR-embedding-2,qwen-embedding,
openai-embedding}. For each embedding method, we executed \wwj and \bas 
100 times with an Oracle budget of 1,000,000.
In Figure \ref{fig:sensi-embedding}, we demonstrate the performance of \bas, 
\wwj, and \uniform with various embeddings on the \dcompany dataset with an 
Oracle budget of 1,000,000 and a maximum blocking ratio of 20\%. As shown, \bas 
consistently outperforms \uniform and \wwj across all embeddings. 
We find that \bas benefits from improved performance with 
better embeddings.

\input{tex/figures/fig_sensitivity_embedding.tex}
Furthermore, we find that traditional TF-IDF sometimes outperforms 
modern dense embeddings (e.g., OpenAI). This counter-intuitive 
result is caused by the specific join semantics of the Company dataset, which
requires high \textit{lexical precision} (e.g., distinguishing a radio company 
from a television company), a task where traditional methods excel by strictly 
penalizing token mismatches. In contrast, dense embeddings prioritize 
\textit{semantic relatedness}, often clustering non-matching companies from the 
same industry, creating extra false positives (e.g., 99\% FPR for OpenAI vs. 
27\% for TF-IDF).\footnote{We measured false positive rates at a fixed recall 
of 50\%.} Therefore, traditional methods or embedding models with instructional 
prompts (e.g., \texttt{gte-Qwen2} \cite{qwen-embedding} and \texttt{stella\_en} 
\cite{stella-embedding}) work better on this dataset. Nevertheless, \bas provides 
statistical guarantees regardless of the embedding's quality.

\input{tex/figure_sensi-B.tex}

\minihead{Impact of the Maximum Blocking Ratio ($\alpha$)} 
We executed \bas with maximum blocking ratios ($\alpha$) ranging from 10\% to 
30\% on the \dflickr datasets. For each value of $\alpha$, we repeated 
experiments 100 times and measured the relative RMSE. As shown in Figure 
\ref{fig:sensi-alpha}, the RMSE of \bas fluctuates by 9.5-14.2\% compared with the 
optimal one. Despite the fluctuation, \bas consistently outperforms baselines, 
demonstrating that \bas is insensitive to $\alpha$.

\minihead{Impact of the Weight Exponents}
We executed \wwj and \bas on the \dcompany dataset with different weight 
exponents. For example, the exponent=0.5 means that the sampling weight and 
budget allocation is proportional to the square root of the similarity score. 
For each exponent, we repeated experiments 100 times and measure the relative RMSE. As shown in 
Figure~\ref{fig:sensi-exp}, \bas achieves relatively better performance with
exponent equal to 1 and consistently outperforms baselines.

%% file: tex/figures/fig_guarantees.tex
\begin{figure}[t!]
    \begin{subfigure}{0.24\textwidth}
        \includegraphics[width=\linewidth]{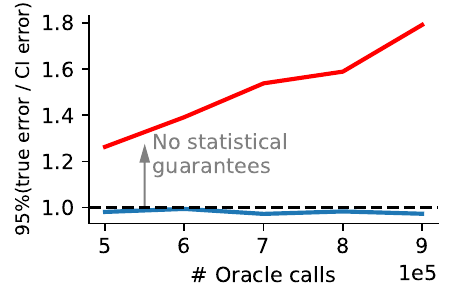}
        \caption{\acount (Roxford)} \label{fig:guarantee-count}
    \end{subfigure}\hspace*{\fill}
    \begin{subfigure}{0.24\textwidth}
        \includegraphics[width=\linewidth]{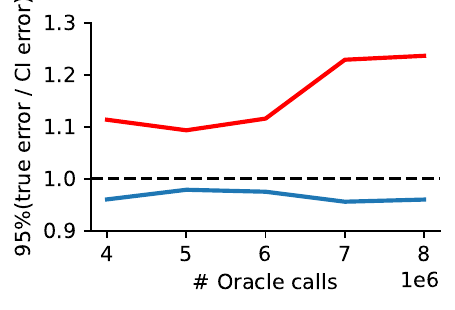}
        \caption{\aavg (Quora)} \label{fig:guarantee-avg}
    \end{subfigure}
    \medskip
    \begin{subfigure}{0.24\textwidth}
        \includegraphics[width=\linewidth]{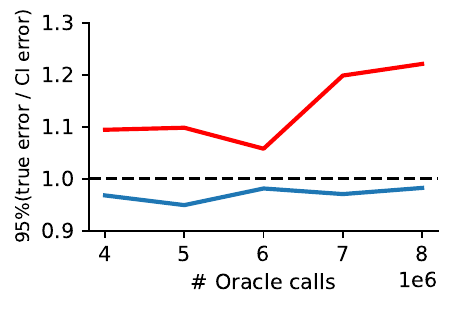}
        \caption{\asum (\dwebmasters)} \label{fig:guarantee-sum}
    \end{subfigure}\hspace*{\fill}
    \begin{subfigure}{0.24\textwidth}
    \includegraphics[width=\linewidth]{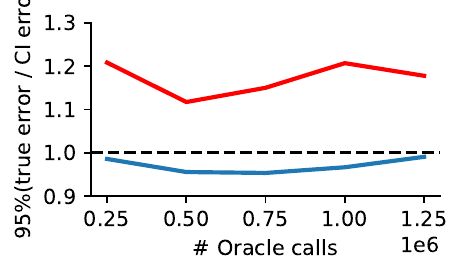}
    \caption{\amedian (Ecomm-Q9)} \label{fig:guarantee-median}
    \end{subfigure}

    \caption{\bas achieves valid statistical guarantees (Error Ratio $\le 1$) across all evaluated aggregators, whereas standard \block consistently produces invalid CI.}
    \label{fig:guarantees}
\end{figure}

%% file: tex/figures/fig_low_guarantees.tex
\begin{figure}[t!]
    \begin{subfigure}{0.24\textwidth}
        \includegraphics[width=\linewidth]{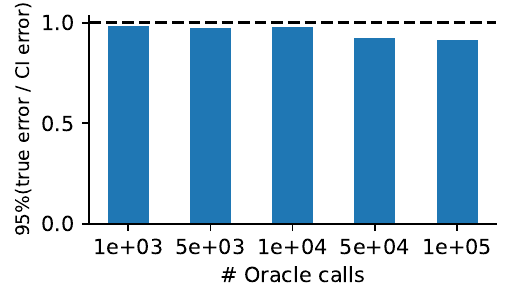}
    \end{subfigure}\hspace*{\fill}
    \begin{subfigure}{0.24\textwidth}
        \includegraphics[width=\linewidth]{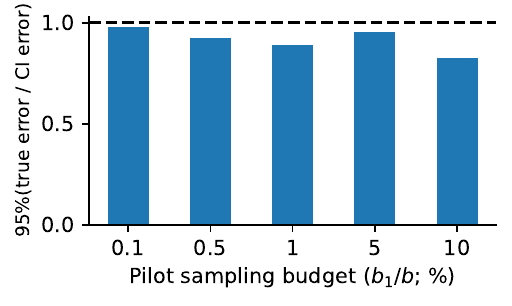}
    \end{subfigure}
    \caption{\bas maintains valid CIs even under 
    extreme conditions, including low Oracle budgets (e.g., 1,000) and small pilot sample sizes (e.g., 0.1\%).}
    \label{fig:guarantee-budget}
\end{figure}

%% file: tex/fig_rrmse.tex
\begin{figure*}[t]
    \begin{subfigure}{0.24\textwidth}
        \includegraphics[width=\linewidth]{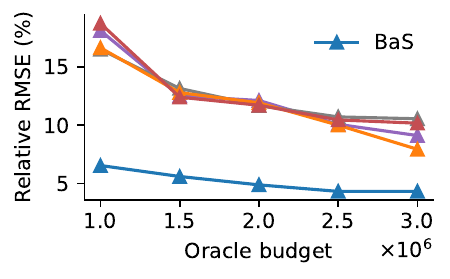}
        \caption{\dcompany-\acount} \label{fig:rrmse-company}
    \end{subfigure}\hspace*{\fill}
    \begin{subfigure}{0.24\textwidth}
        \includegraphics[width=\linewidth]{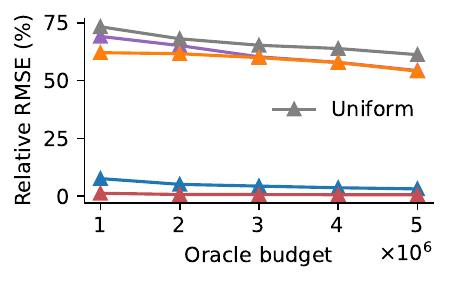}
        \caption{\dquora-\aavg} \label{fig:rrmse-quora}
    \end{subfigure}\hspace*{\fill}
    \begin{subfigure}{0.24\textwidth}
        \includegraphics[width=\linewidth]{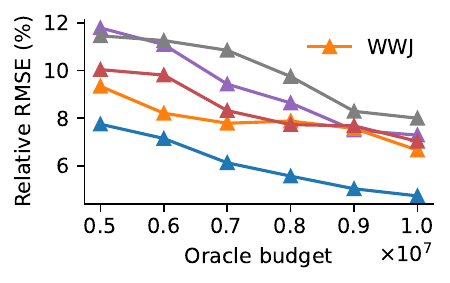}
        \caption{\dwebmasters-\asum} \label{fig:rrmse-webmaster}
    \end{subfigure}\hspace*{\fill}
    \begin{subfigure}{0.24\textwidth}
    \includegraphics[width=\linewidth]{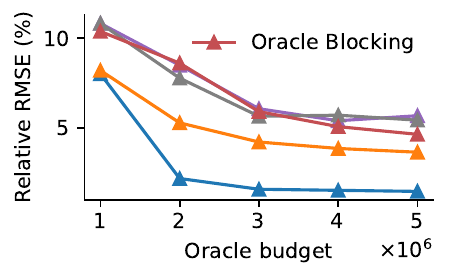}
    \caption{\droxford-\acount} \label{fig:rrmse-roxford}
    \end{subfigure}
    \medskip

    \begin{subfigure}{0.24\textwidth}
        \includegraphics[width=\linewidth]{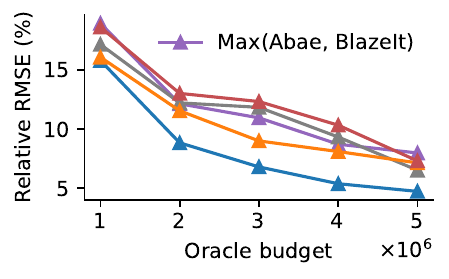}
        \caption{\dflickr-\acount} \label{fig:rrmse-flickr30k}
    \end{subfigure}\hspace*{\fill}
    \begin{subfigure}{0.24\textwidth}
        \includegraphics[width=\linewidth]{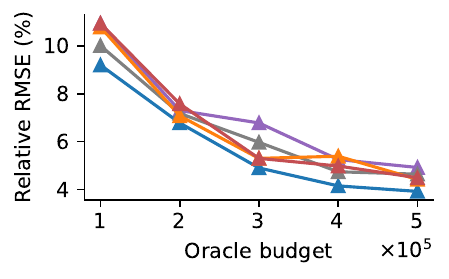}
        \caption{\dveri-\aavg} \label{fig:rrmse-veri}
    \end{subfigure}\hspace*{\fill}
    \begin{subfigure}{0.24\textwidth}
        \includegraphics[width=\linewidth]{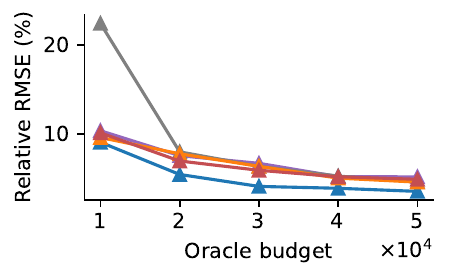}
        \caption{Ecomm-Q7-\acount} \label{fig:rrmse-eq7}
    \end{subfigure}\hspace*{\fill}
    \begin{subfigure}{0.24\textwidth}
        \includegraphics[width=\linewidth]{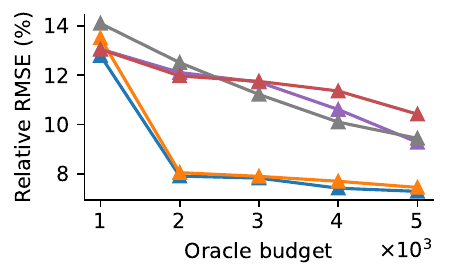}
        \caption{Movie-Q5-\aavg} \label{fig:rrmse-mq5}
    \end{subfigure}
    \medskip

    \begin{subfigure}{0.24\textwidth}
        \includegraphics[width=\linewidth]{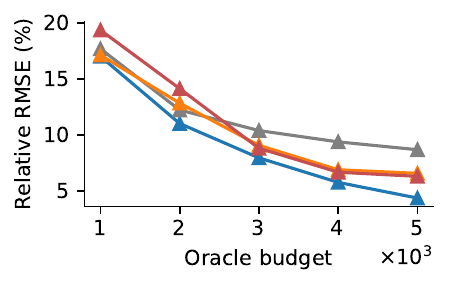}
        \caption{Movie-Q6-\amax} \label{fig:rrmse-mq6}
    \end{subfigure}\hspace*{\fill}
    \begin{subfigure}{0.24\textwidth}
        \includegraphics[width=\linewidth]{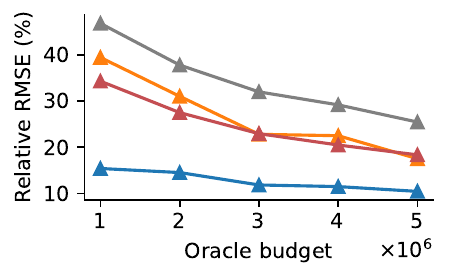}
        \caption{Ecomm-Q8-\amin} \label{fig:rrmse-eq8}
    \end{subfigure}\hspace*{\fill}
    \begin{subfigure}{0.24\textwidth}
        \includegraphics[width=\linewidth]{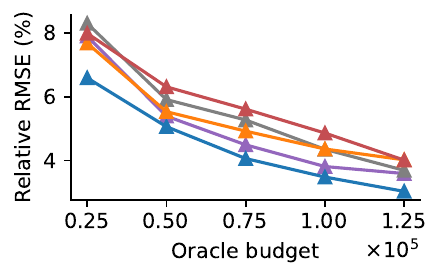}
        \caption{Ecomm-Q9-\amedian} \label{fig:rrmse-eq9}
    \end{subfigure}\hspace*{\fill}
    \begin{subfigure}{0.24\textwidth}
        \includegraphics[width=\linewidth]{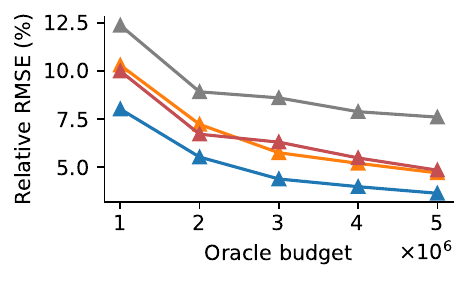}
        \caption{Ecomm-Q10a-\agroupby\ (3 tables)} \label{fig:rrmse-eq10a}
    \end{subfigure}
    \medskip

    \begin{subfigure}{0.24\textwidth}
        \includegraphics[width=\linewidth]{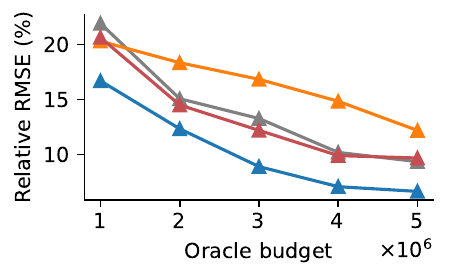}
        \caption{Ecomm-Q10-\acount (3 tables)} \label{fig:rrmse-eq10}
    \end{subfigure}\hspace*{\fill}
    \begin{subfigure}{0.24\textwidth}
        \includegraphics[width=\linewidth]{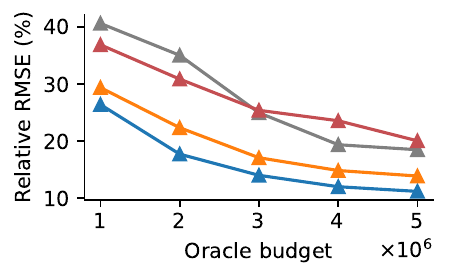}
        \caption{Ecomm-Q11-\acount (4 tables)} \label{fig:rrmse-eq11}
    \end{subfigure}\hspace*{\fill}
    \begin{subfigure}{0.24\textwidth}
        \includegraphics[width=\linewidth]{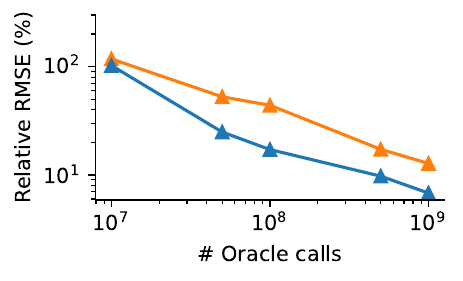}
        \caption{\dcompany-Scale-\acount (6 tables)} \label{fig:rrmse-companyL}
    \end{subfigure}\hspace*{\fill}
    \begin{subfigure}{0.24\textwidth}
        \includegraphics[width=\linewidth]{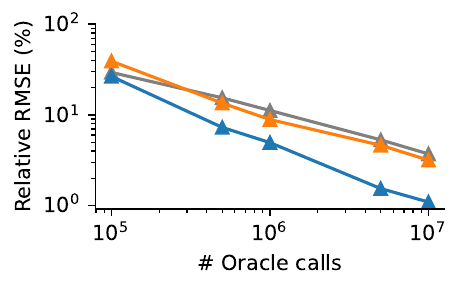}
        \caption{\droxford-Scale-\acount} \label{fig:rrmse-roxfordl}
    \end{subfigure}

    \caption{BaS consistently reduces RMSE across a diverse workload. It outperforms 
    baselines in linear aggregators (\ref{fig:rrmse-company}--\ref{fig:rrmse-mq5}), extreme aggregators (\ref{fig:rrmse-mq6} and 
    \ref{fig:rrmse-eq8}), \amedian (\ref{fig:rrmse-eq9}), \agroupby (\ref{fig:rrmse-eq10a}), multi-table join (\ref{fig:rrmse-eq10}--\ref{fig:rrmse-companyL}), 
    and large-scale join queries (\ref{fig:rrmse-companyL} and \ref{fig:rrmse-roxfordl}). BaS remains robust in both rare-event (Low Selectivity) and dense (High Selectivity) scenarios.}
    \label{fig:rrmse}
\end{figure*}

%% file: tex/figures/fig_extension.tex
\begin{figure}
\centering
\begin{subfigure}{0.23\textwidth}
        \centering
        \includegraphics[width=\linewidth]{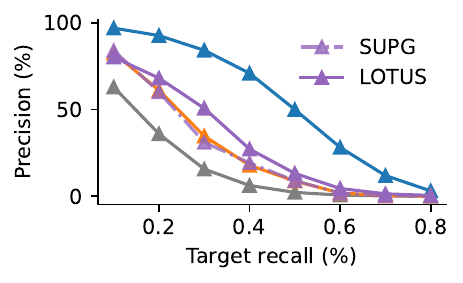}
        \caption{\bas achieves up to 43\% higher precision than 
        LOTUS while maintaining the same recalls guarantees.} \label{fig:recall}
    \end{subfigure}\hspace*{\fill}
    \begin{subfigure}{0.23\textwidth}
        \centering
        \includegraphics[width=\linewidth]{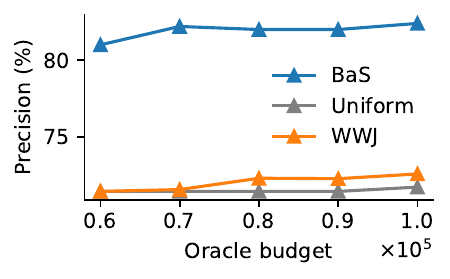}
        \caption{\bas improves precision on \atopk queries by up to 10\% over baselines, while LOTUS does not support \atopk with guarantees} \label{fig:topk}
    \end{subfigure}
    \caption{\bas extends to improve selection queries.}
\end{figure}


%% file: tex/figures/fig_join_opt.tex
\begin{figure}[t]
    \centering
    \includegraphics[width=0.8\linewidth]{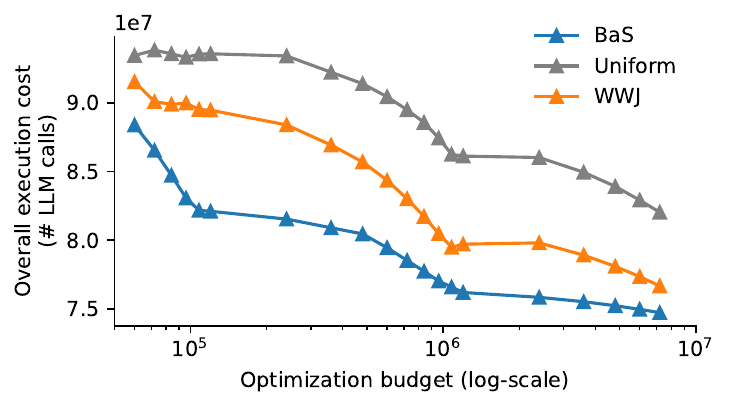}
    \caption{Application of approximate \acount to join optimization: By 
    providing accurate cardinality estimates, BaS enables the query optimizer to 
    select better join orders, reducing the execution cost of Ecomm-Q11 by 
    12.7\%.}
    \label{fig:join_opt}
\end{figure}

%% file: tex/figures/fig_ablation.tex
\begin{figure}[t!]
    \centering
    \includegraphics[width=0.8\linewidth]{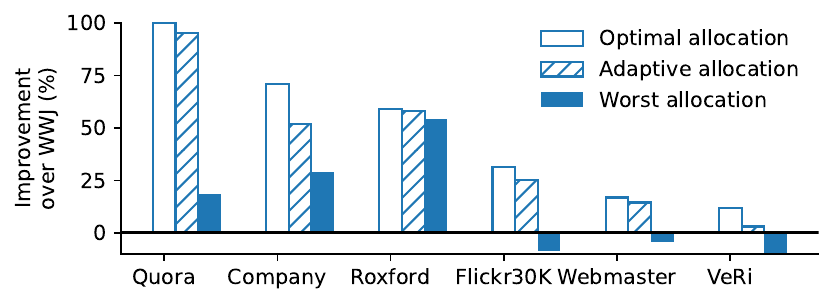}
    \caption{\bas determines near-optimal allocations.} 
    \label{fig:optimal}
\end{figure}

%% file: tex/figures/fig_sensitivity_embedding.tex
\begin{figure}[t!]
    \centering
    \includegraphics[width=0.82\linewidth]{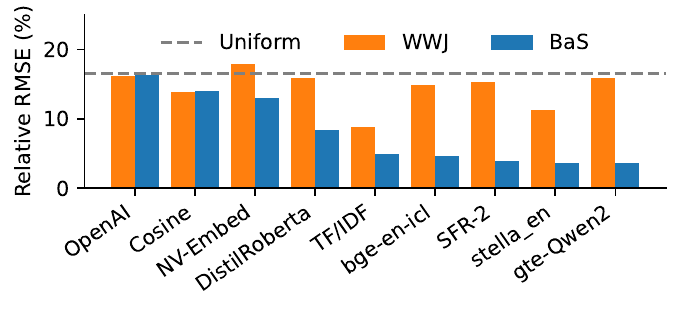}
    \caption{\bas consistently outperforms \uniform and \wwj across various 
    types of embeddings.} 
    \label{fig:sensi-embedding}
\end{figure}

\begin{figure}
    \centering
    \includegraphics[width=0.82\linewidth]{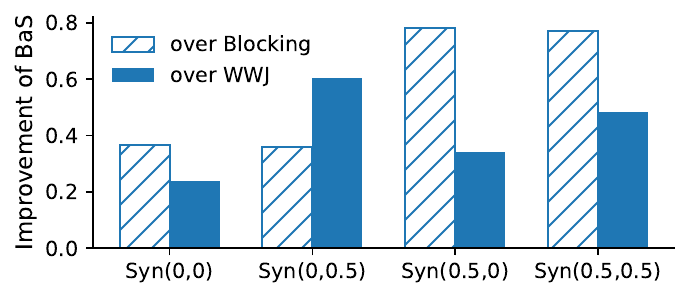}
    \caption{On the synthstic dataset \texttt{Syn}(FN, FP), \bas outperforms 
    baselines even as embedding quality degrades. \bas is particularly effective 
    compared to \block when False Negatives (FN) are high, and compared to \wwj 
    when False Positives (FP) are high.}
    \label{fig:synthetic}
\end{figure}

%% file: tex/figure_sensi-B.tex
\begin{figure}[t!]
    \centering
    \begin{subfigure}{0.23\textwidth}
        \includegraphics[width=\linewidth]{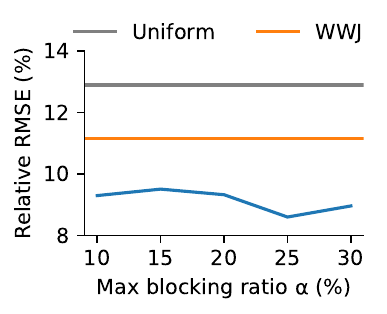}
        \caption{$\alpha$ vs relative RMSE for the \dwebmasters dataset.} 
        \label{fig:sensi-alpha}
    \end{subfigure}
    \hfill
    \begin{subfigure}{0.227\textwidth}
        \includegraphics[width=\linewidth]{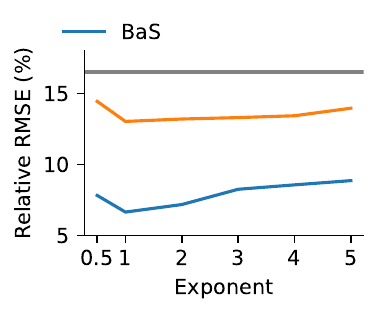}
        \caption{Weight exponent vs relative RMSE for the \dcompany dataset.} 
        \label{fig:sensi-exp}
    \end{subfigure}
    \caption{Sensitivity of \bas in terms of $\alpha$ and the weight exponents.
    \bas consistently outperforms baselines.}
\end{figure}

%% file: tex/related_work.tex
\section{Related Work}\label{sec:rel-work}
We draw inspiration from classic join optimization algorithms and leverages 
recent advancements in EM and AQP. We summarize the relevant literature on these areas.

\minihead{Optimization of approximate join queries} Joins are fundamental but 
notoriously expensive operations in relational database management systems 
\cite{db-join-general,rdb-join-general}. To optimize and accelerate join queries 
with structured data, various algorithms and systems are developed in both 
offline and online scenarios. 

Prior work developed offline sampling algorithms to address the challenge that 
joining uniform samples does not constitute an independent sample of the join 
\cite{join-sample1,join-sample2,structured-join-sampling}. Specifically, 
researchers proposed the Stratified-Universe-Bernoulli sampling that achieves 
the theoretical optimal sampling performance within a constant factor 
\cite{structured-join-sampling}. For joins with specific types, prior work 
proposed precomputed synopses and histogram method 
\cite{join-histogram,structured-join-synopsis}. To accelerate join queries in 
the online scenario, prior work proposed ripple join \cite{ripplejoin}. Wander 
Join improved the sampling efficiency using indices and random walks 
\cite{structured-join-index}. 

Recently, for unstructured data, researchers have focused on efficient 
implementations of the semantic join operator \cite{patel2024lotus,
trummer2025implementing}. Notably, \citet{trummer2025implementing} proposes a 
batch-oriented block nested loop algorithm that optimizes prompt construction to 
evaluate multiple row pairs in a single LLM invocation. This line of work 
addresses the \textit{physical execution efficiency} of the join operator. In 
contrast, \bas addresses the logical approximate execution of such queries via 
reducing variance. These approaches are highly complementary. \bas optimizes 
which pairs to check, while \cite{trummer2025implementing} reduces the cost of 
checking those pairs, particularly in the dense blocking regime.

\minihead{Entity Matching} EM seeks to identify common objects across 
datasets \cite{em-survey1,blocking1}. Earlier approaches have modeled EM 
problems as similarity joins \cite{similarity-join,em-similarity}. However, EM 
is prohibitively expensive due to pair-wise comparisons \cite{em-unstructured}. 
The blocking method is developed to reduce the number of pair-wise comparisons. Traditional 
blocking methods relied on heuristics and expert knowledge to develop rule-based 
block constructions, such as suffix arrays \cite{em-expert1,em-expert2}, sorted 
neighborhood \cite{em-expert3,em-expert4}, and sorted block \cite{em-expert5}.

Recent advances in ML models, especially LLMs, have facilitated learning-based 
methods in EM systems 
\cite{deep-em,dataset-company-work2,dataset-company-work1,blocking2}. Existing 
work proposed to use language models for textual data embedding, classification 
\cite{deep-em,dataset-company-work2,dataset-company-work1}, and blocking 
\cite{blocking2}. Additionally, recent work improved the robustness of 
learning-based methods using data augmentation 
\cite{em-learning-robust}. 

\minihead{AQP for unstructured data} Recent work addressed AQP for unstructured 
data using machine learning models. To reduce the cost of executing expensive 
machine learning models, prior work leveraged approximate indexes \cite{focus}, 
binary classifiers \cite{probabilistic-predicates}, and specialized models 
\cite{noscope} for AQP without statistical guarantees. Furthermore, 
advancements including query-driven object tracking algorithms and 
segmentation-based tracking queries, further reduced the query cost. To provide guarantees on accuracy, existing work used proxy models with 
control variate to process aggregation and limit queries on video frames 
\cite{blazeit}, importance sampling to process selection queries \cite{supg}, 
and stratified sampling to process aggregation queries \cite{abae}. 

Join operations have been commonly used for unstructured data in prior AQP 
systems for unstructured data \cite{jo2024thalamusdb,liu2024optimizing,
patel2024lotus}. However, prior work either focused on join with structured join 
keys \cite{jo2024thalamusdb,liu2024optimizing}, which is orthogonal to our work, 
or applies calibration techniques without statistical guarantees 
\cite{patel2024lotus}.

%% file: tex/conclusion.tex
\section{Conclusion}
In this work, we propose Weighted Wander Join and Blocking-augmented Sampling to 
approximately process analytical join queries over unstructured data with 
statistical guarantees. We address the inaccuracy of embedding models by 
integrating sampling and blocking, and propose an adaptive allocation algorithm 
to minimize the errors. Theoretically, we prove the optimality and 
superiority of \bas. Our evaluation results demonstrate that \bas outperforms 
baselines by up to 21$\times$, saving costs by up to 188$\times$ on 
real-world datasets. These results highlight the potential of embeddings and 
cost-effective algorithms in analyzing unstructured data with joins. 

\section{Acknowledgements}
We are grateful to the CloudLab for providing computing resources
for experiments \cite{Duplyakin}. We thank Kaimeng Zhu and Siheng Pan for their
help during the initial implementation and experimentation of this work. This 
research was supported in part by Google.

%% file: tex/appendix/cost_analysis.tex
\section{Cost and Latency Analysis} \label{subsec:runtime}
\begin{figure}[t]
        \centering
        \includegraphics[width=0.6\linewidth]{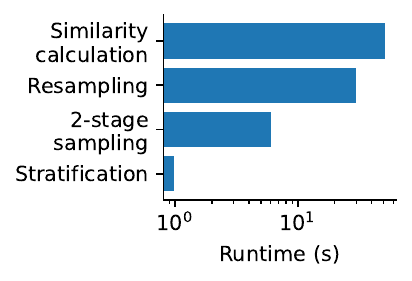}
        \caption{Decomposition of the average (log-scale) latency for CPU computation of \bas.}
        \label{fig:latency}
    \end{figure}
In this section, we analyze the cost of executing ML models and the latency of
our algorithm. We use \dquora as an example to show that the cost of executing
embedding models is minimal relative to running the Oracle. We observe 
similar results on other datasets. Consider OpenAI's text-embedding-3-large 
model, the most expensive among the embedding models evaluated in Section 
\ref{subsec:sensitivity}. It costs \$0.05 to generate embeddings for all data records 
using the batch API. In contrast, labeling all data pairs using the 
state-of-the-art prompt template \cite{em-llm1} would cost \$78,848 with GPT-4o 
or \$4,730 with GPT-4o mini using the batch API. This is about six orders of 
magnitude more expensive than generating embeddings.

We then analyze the latency of the CPU computations by measuring the 
single-thread runtime of \bas on Intel 8562Y+ with 200 GB memory. As shown in 
Figure~\ref{fig:latency}, similarity calculation and resampling contribute to 
95\% of the latency. The overall latency on CPU is lower than 90 seconds, which is 
minimal compared to executing the Oracle on GPUs or relying on human labelers.

%% file: tex/appendix/notations.tex
\begin{table*}[h]
    \centering
    \normalsize
    \caption{Summary of Notation}
    \begin{tabular}{ll}
        \toprule
        Symbol & Description \\
        \midrule
        $D$                                     
        & stratified dataset of data tuples in the cross product of join tables \\
        
        $W$                                     
        & normalized similarity scores of tuples in $D$ \\
        
        $O$                                     
        & Oracle \\
        
        $g$                                     
        & attribute where the aggregate is computed \\
        
        $b, b_1, b_2$                           
        & Oracle budget of the overall procedure, pilot sampling, and sampling+blocking execution \\
        
        $p$                                     
        & confidence \\
        
        $\alpha$                                
        & maximum blocking ratio \\
        
        $K$                                     
        & number of strata \\
        
        $\widehat\acount, \widehat\asum, \widehat\aavg, \widehat{\texttt{AGG}}$   
        & \makecell[l]{estimated aggregate for the entire dataset, sampling regime (with a subscript $s$), \\ and blocking regime (with a subscript $b$)}\\
        
        $l, u$                                  
        & lower, upper bound of the confidence interval \\
        
        $\mathcal{O}(\cdot)$                    
        & Bachmann-Landau big-O notation \\
        
        $\beta, \beta^*, \hat\beta^*$                        
        & the allocated strata to blocking that is given, optimal, and estimated \\
        
        $n^{(1)}_i, n_i$                        
        & assigned Oracle budget for stratum $i$ during pilot sampling and blocking+sampling execution \\
        
        $\sigma_i^2, \hat\sigma_i^2$                        
        & (estimated) sampling variance of stratum $i$ \\
        
        $\mu, \hat\mu$                        
        & (estimated) expectation of stratum $i$ \\
        
        $S^{(1)}_i, S_i$                        
        & sampled data tuples from stratum $i$ during pilot sampling and blocking+sampling \\
        
        $t_j$                                   
        & t-statistic of the $j$-th resampling iteration \\
        
        $\tilde{D}, \tilde{D}^{(s)}$
        & the set of matching tuples of the entire dataset or the sampling regime \\

        \bottomrule
    \end{tabular}
\end{table*}

%% file: tex/appendix/setup.tex
\subsection{Setup}
In this section, we describe \bas in detail as a setup for the theoretical 
analysis. We will go over the entire procedure of \bas, including stratification,
pilot sampling, blocking, sampling, and resampling.

\minihead{Inputs and Outputs} 
\bas takes as inputs the cross product of tables $D$, the similarity scores $W$, 
the Oracle $O$, the overall Oracle budget $b$ divided for pilot sampling ($b_1$) 
and bloking+sampling execution ($b_2$), the confidence $p$, the maximum blocking 
ratio $\alpha$, and the number of strata $K$. \bas outputs an estimated aggregate 
$\hat\mu$ and its confidence interval $[l, u]$.

\minihead{Stratification}
We divide $D$ into a maximum blocking regime and a minimum sampling regime 
($D_0$). The maximum blocking regime contains tuples with the top 
$\alpha \cdot b$ similarity scores, where $\alpha$ is a parameter between 0 and 
1 to control the size of the maximum blocking regime. Next, we stratify the 
maximum blocking regime into $K$ strata ($D_1, \ldots, D_K$) with equal sizes. 
The number of strata $K$ is automatically determined to ensure that each stratum 
has an Oracle budget of at least 1,000. 

\minihead{Pilot Sampling}
For each stratum, we execute \wwj to obtain a pilot sample $S_i^{(1)}$. The 
Oracle budget for the stratum $i$ in the pilot sampling is calculated as
\begin{equation}
    n^{(1)}_i = b_1 \cdot \frac{ \sum_{s \in D_i} W(s)}{\sum_{s \in D} W(s)} \label{eq:sample-size}
\end{equation}
We can estimate the sampling variance as
\begin{equation*}
    \hat\sigma^2_i = \frac{1}{n^{(1)}-1} \sum_{s \in S^{(1)}_i} 
    \left(\frac{g(s)O(s)}{W(s)|D_i|}-\left(\frac{1}{n_i^{(1)}}\sum_{s' \in S^{(1)}_i} \frac{g(s')O(s')}{W(s')|D_i|}\right)\right)^2 W(s)
\end{equation*}
Given an allocation $\beta$, we can then estimate the MSE of \bas for a \asum 
aggregate as follows
\begin{equation*}
    \widehat{MSE}_\asum(D, \beta, W, b_2) = \sum_{\substack{0 \le i\le k, i\notin \beta}} \frac{|D_i|^2}{n_i} \cdot \hat\sigma_i^2
\end{equation*}
where $n_i$ is the assigned Oracle budget for stratum $i$ in the stage of 
blocking+sampling execution, calculated as
\begin{equation*}
    n_i = 
     \begin{cases}
        (b_2 - \sum_{j \in \beta} |D_i|) \cdot \frac{\sum_{s \in D_i} W(s)}{\sum_{1\le j\le K, j\notin \beta}\sum_{s \in D_j} W(s)}         & i \notin \beta \\
        |D_i|    & i \in \beta \\
     \end{cases}
\end{equation*}
Next, we obtain the estimated optimal allocation by solving the following 
optimization problem using iterative methods.
\begin{equation*}
    \hat\beta^* = \mathop{\arg\min}_{\beta \subset \{1, \ldots, K\}} \widehat{MSE}_{\asum}(D, \beta, W, b_2)
\end{equation*}

\minihead{Blocking+Sampling}
Given the optimal allocation $\hat\beta^*$, we use the Oracle budget $b_1$ to 
execute the Oracle on the strata that are allocated to the blocking regime. 
On the remaining strata, we execute \wwj to obtain a sample using the remaining 
Oracle budget. We can obtain the sample $S_i$ on the stratum $i$ that is 
allocated to sampling. Then, we merge the result of all strata and the results 
of pilot sampling to estimate the aggregate. Specifically, we estimate the 
aggregate as follows:
\begin{align*}
    & \widehat\acount_s = \frac{1}{\sum_{i \notin \hat\beta^*} n_i} \sum_{i \notin \hat\beta^*}\sum_{s \in S_i} \frac{O(s)}{W(s) |D_i|},\ 
    \acount_b = \sum_{i\in \hat\beta^*} \sum_{s \in D_i} O(s) \\
    & \widehat\acount_s = \frac{1}{\sum_{i \notin \hat\beta^*} n_i} \sum_{i \notin \hat\beta^*}\sum_{s \in S_i} \frac{g(s) O(s)}{W(s) |D_i|},\ 
    \acount_b = \sum_{i\in \hat\beta^*} \sum_{s \in D_i} g(s) O(s) \\
    & \widehat\acount = \widehat\acount_s + \acount_b,\  \widehat\asum = \widehat\asum_s + \asum_b \\
    & \widehat\aavg = \left(\widehat\asum_s + \asum_b\right)\Big/\left(\widehat\acount_s + \acount_b\right)
\end{align*}

\minihead{Resampling}
We apply the bootstrap-t resampling scheme to calculate the CI 
\cite{boostrap-comparison}. The bootstrap-t scheme estimates the standard error 
$\frac{\hat\mu - \mu}{\sigma}$ (\ie, t-statistic) of the underlying distribution 
by resampling existing samples. To process aggregation queries, we first 
calculate the mean and standard deviation of the estimator. Next, we use 
sampling with replacement to resample from all existing samples 
($S^{(1)} \cup S$). We calculate the t-statistic of the $j$-th iteration as 
follows.
\begin{equation*}
    t_j = \frac{\widehat{\texttt{AGG}}_j - \widehat{\texttt{AGG}}}{\hat\sigma_j}
\end{equation*}
where $\widehat{\texttt{AGG}}_j$ and $\hat\sigma_j$ are the estimated aggreagte
and standard deviation on the $j$-th resample. To achieve statistical guarantees, 
we repeat the resampling a sufficient number of times (\eg, 1000). Finally, we 
use the percentiles of resampled t-statistics to construct the CI, that is
\begin{equation*}
    l = \widehat{\texttt{AGG}} - \mathrm{Percentile}\left(t, \frac{1-p}{2}\right), \quad
    u = \widehat{\texttt{AGG}} - \mathrm{Percentile}\left(t, 1 - \frac{1-p}{2}\right)
\end{equation*}

%% file: tex/appendix/optimal.tex
\subsection{\bas Converges to the Optimal Allocation} \label{sec:app-optimal}
\begin{theorem} \label{theom:opt}
The MSE with the estimated minimizer $\hat\beta^*$ converges to that with the 
true minimizer $\beta^*$ with the rate $\mathcal{O}\left(1/\sqrt{b_1}\right)$:
\begin{equation}
    \frac{MSE(D, \widehat\beta^*, W, b_2) - 
    MSE(D, \beta^*, W, b_2)}{MSE(D, \beta^*, W, b_2)} 
    = \mathcal{O}\left(\frac{1}{\sqrt{b_1}}\right) \notag
\end{equation}
\end{theorem}
\begin{proof}
We first show the estimated MSE converges to the true MSE. To achieve that, we
establish the convergence rate of the sample mean $\hat\mu_i$ and variance 
$\hat\sigma_i^2$ for stratum $i$ because MSE estimations are calculated with 
summations, multiplications, and divisions of basic estimators $\hat\mu_i, 
\hat\sigma_i^2$.

We construct the concentration inequality for the sample mean $\hat\mu_i$ and 
sample variance $\hat\sigma_i^2$ of each stratum $D_i, i=0, \ldots, K$. Given 
the i.i.d. sample, we can bound the relative errors of $\hat\mu_i$ and 
$\hat\sigma_i^2$ with Chebeyshev's Inequality, where the variance of an i.i.d. 
sample variance is calculated as 
$Var\left[\hat\sigma_i^2\right] = \frac{2\sigma_i^4}{n_i-1}$
\cite{stats-infer}.
\begin{align}
    \mathbb{P}\left[
        \left|
            \frac{\mu_i - \hat\mu_i}{\mu_i}
        \right| \le \frac{\sigma_i}{\mu_i\sqrt{\delta n_i^{(1)}}}
    \right] & \ge 1 - \frac{\delta}2 \\
    \mathbb{P}\left[
        \left|\frac{\sigma_i^2 - \hat\sigma_i^2}{\sigma_i^2}\right| 
        \le \sqrt{\frac{2}{\delta(n_i^{(1)}-1)}}
    \right] & \ge 1 - \frac{\delta}2
\end{align}
where $\delta$ is a small probability.

Since $n_i^{(1)}$ is linearly related to the Oracle budget $b_1$ (Eq. 
\ref{eq:sample-size}), we have the following big-$\mathcal{O}$ notations showing 
the convergence rate that holds with high probability.
\begin{equation}
    \frac{\mu_i - \hat\mu_i}{\mu_i} = \mathcal{O}\left(b_1^{-1/2}\right), \quad
    \frac{\sigma_i^2 - \hat\sigma_i^2}{\sigma_i^2} 
        = \mathcal{O}\left(b_1^{-1/2}\right) \label{eq:rerror}
\end{equation}

Next, we show the convergence rate of relative error is preserved after 
summations, multiplications, and divisions. Namely, for unbiased estimators 
$\mu_1$ and $\mu_2$ with relative error converging to 0 with a rate of 
$\mathcal{O}\left(b_1^{-1}\right)$ ($t > 0$), we have Equations 
\ref{eq:error-sum} - \ref{eq:error-div}.
\begin{align}
    (\mu_1+\mu_2) - (\hat\mu_1+\hat\mu_2) 
    &= \mu_1 \cdot \mathcal{O}\left(b_1^{-1}\right) 
        + \mu_2 \cdot \mathcal{O}\left(b_1^{-1}\right) \notag \\
    &= (\mu_1+\mu_2)\cdot \mathcal{O}\left(b_1^{-1}\right) \label{eq:error-sum} \\
    \mu_1\mu_2 - \hat\mu_1\hat\mu_2 
    &= \mu_1\mu_2 - \mu_1\left(1+\mathcal{O}\left(b_1^{-1}\right)\right)
        \mu_2\left(1+\mathcal{O}\left(b_1^{-1}\right)\right)\notag \\
    &= \mu_1\mu_2\left(\mathcal{O}\left(b_1^{-1}\right)
        + \mathcal{O}\left(b_1^{-2t}\right)\right) \notag\\
    &= \mu_1\mu_2\mathcal{O}\left(b_1^{-1}\right) \label{eq:error-multi}\\
    \mu^{-1} - \hat\mu^{-1} &= \mu^{-1}\hat\mu^{-1}(\hat\mu - \mu) \notag\\
    &= \mu^{-2}\left(1+\mathcal{O}\left(b_1^{-1}\right)\right)^{-1}
        \mu\mathcal{O}\left(b_1^{-1}\right) \notag \\
    &= \mu^{-1}\mathcal{O}\left(b_1^{-1}\right) \label{eq:error-div}
\end{align}

Therefore, given the convergence rate of basic estimators (Eq. \ref{eq:rerror}) 
and propagation rules (Eq. \ref{eq:error-sum}-\ref{eq:error-div}), the relative 
error of our optimization objective converges to 0 with the same rate 
$\mathcal{O}\left(b_1^{-1/2}\right)$, with high probability. Namely, the following 
probabilistic bound holds for any $\beta \subset \{1, \ldots, K\}$ with a high 
probability of $1-\delta/2$
\begin{equation}
    \mathbb{P} \left[
        \left|
            \frac{MSE_{\texttt{AGG}}(D, \beta, W, b_2) 
                - \widehat{MSE}_{\texttt{AGG}}(D, \beta, W, b_2)}
            {MSE_{\texttt{AGG}}(D, \beta, W, b_2)}
        \right| \le \frac{C}{\sqrt{b_1}}
    \right] \ge 1 - \frac{\delta}{2} 
    \label{eq:mse-prob-bound}
\end{equation}
where $C$ is a constant independent of $n$. Based on Equation 
\ref{eq:mse-prob-bound}, we can derive the following bound for MSE that holds 
with a high probability of $1-\delta/2$.
\begin{equation*}
    \frac{1}{1 + \frac{C}{\sqrt{n}}} \widehat{MSE}_{\texttt{AGG}}(\beta) \le 
        MSE_{\texttt{AGG}}(\beta) \le 
        \frac{1}{1 - \frac{C}{\sqrt{n}}} \widehat{MSE}_{\texttt{AGG}}(\beta)
\end{equation*}
where we omit the parameters $D, W, b_2$ for simplicity.

We then derive the upper bound of the difference between the MSE of \bas and the 
optimal MSE. Given the estimated minimizer $\hat\beta^*$ and the true minimizer 
$\beta^*$, we can establish the following upper bound of the difference between 
$MSE_{\texttt{AGG}}(\hat\beta^*)$ and $MSE_{\texttt{AGG}}(\beta^*)$.
\begin{align}
    & MSE_{\texttt{AGG}}(\hat\beta^*) - MSE_{\texttt{AGG}}(\beta^*) 
    \le \frac{1}{1 - \frac{C_1}{\sqrt{n}}} 
        \widehat{MSE}_{\texttt{AGG}}(\hat\beta^*) 
        - MSE_{\texttt{AGG}}(\beta^*) \label{step:bound1}\\
    &\le \frac{1}{1 - \frac{C_1}{\sqrt{n}}} 
        \widehat{MSE}_{\texttt{AGG}}(\beta^*) 
        - MSE_{\texttt{AGG}}(\beta^*) \label{step:def}\\
    &\le \frac{1}{1 - \frac{C_1}{\sqrt{n}}}
        \left( 1 + \frac{C_2}{\sqrt{n}}\right) 
        MSE_{\texttt{AGG}}(\beta^*) 
        - MSE_{\texttt{AGG}}(\beta^*) \label{step:bound2}\\
    & =  \frac{C_1+C_2}{\sqrt{n} - C_1} MSE_{\texttt{AGG}}(\beta^*) 
\end{align}
where inequalities \ref{step:bound1} and \ref{step:bound2} apply the probabilistic inequalities of $MSE_{\texttt{AGG}}(\hat\beta^*)$ and $MSE_{\texttt{AGG}}(\beta^*)$
respectively, inequality \ref{step:def} is due to the definition of estimated minimizer $\hat\beta^*$, and $C_1, C_2$ are constants from the probabilistic
inequalities.

Finally, we derive the upper bound of the relative error of the optimization 
objective that holds with high probability.
\begin{equation*}
    \frac{MSE_{\texttt{AGG}}(\hat\beta^*) 
        - MSE_{\texttt{AGG}}(\beta^*)}{MSE_{\texttt{AGG}}(\beta^*)} 
    \le \frac{C_1 + C_2}{\sqrt{n}-C_1}
\end{equation*}
This upper bound shows that the relative error between the MSE with estimated 
minimizer $\hat\beta^*$ and the MSE with actual minimizer $\beta^*$ converges to 
0 at the rate $\mathcal{O}\left(b_1^{-1/2}\right)$ with high probability.
\end{proof}

%% file: tex/appendix/comparison.tex
\subsection{\bas Outperforms or Matches \wwj} \label{sec:app-comparison}
\begin{theorem} \label{theo:comp}
    If there exists an allocation $\beta$ such that the following two conditions hold
    \begin{align}
        \mathbb{E}_{s\in \tilde{D}^{(s)}}\left[\frac{1/|\tilde{D}^{(s)}|}{W(s)}\right] 
        &\le \mathbb{E}_{s\in D}\left[\frac{1/|D|}{W(s)}\right]
        \label{eq:cond1}\\
        \frac{|\tilde{D}^{(s)}|^2}{b_2^{(s)}} &\le \frac{|D|^2}{b} \label{eq:cond2}
    \end{align}
    \bas outperforms \wwj asymptotically, \ie,
    \begin{equation*}
        MSE_{\texttt{SUM}} = C \cdot MSE^{(w)}_{\texttt{SUM}} 
            + \mathcal{O}\left(b_1^{-1}b_2^{-1/2}\right) 
    \end{equation*}
    where $C$ is a coefficient less than 1:
    \begin{equation}
        C < \frac{|\tilde{D}^{(s)}|^2 \big/ b_2^{(s)}}{|\tilde{D}|^2 / b} 
            \frac{\mathbb{E}_{s\in \tilde{D}^{(s)}}
                \left[\frac{1/|\tilde{D}^{(s)}|}{W(s)}\right]}
                { \mathbb{E}_{s\in \tilde{D}}\left[\frac{1/|\tilde{D}|}{W(s)}\right]} 
        \le 1 \notag
    \end{equation}
    Otherwise, \bas matches \wwj asymptotically, \ie,
    \begin{equation*}
        MSE_{\texttt{SUM}} \le MSE^{(w)}_{\texttt{SUM}} 
            + \mathcal{O}\left(b_1^{-1}b_2^{-1/2}\right)
    \end{equation*}
\end{theorem}

\begin{proof}
We first derive the ratio between the MSE of \bas and importance sampling, 
given a deterministic allocation $\beta$. The MSe of \bas for a \asum aggregate 
can be calculated as follows.
\begin{align*}
    & MSE_{\asum} \\ \notag 
    & = \sum_{i\notin \beta} \frac{1}{n_i} \left( 
            \mathbb{E}_{s\sim W, s\in D_i}
                \left[ \left( \frac{g(s)O(s)}{W(s)} \right)^2 \right] 
            - \left(
                \mathbb{E}_{s\sim W, s\in D_i}
                \left[ \frac{g(s)O(s)}{W(s)} \right] 
            \right)^2 
        \right)\\
    & = \frac{1}{b_2^{(s)}} \sum_{i\notin \beta} 
        \left( 
            |D_i| \mathbb{E}\left[ \frac{g(s)^2 O(s)}{r_i W(s)} \right] 
            - \frac{1}{r_i} 
                \left( |D_i| \mathbb{E}\left[ g(s)O(s) \right] \right)^2 
        \right)
\end{align*}
where 
\begin{equation*}
    b_2^{(s)} = b_2 - \sum_{i \in \beta} |D_i|, 
    \quad r_i = \frac{\sum_{s \in D_i} W(s)}{\sum_{j\notin \beta} \sum_{s \in D_j} W(s)}
\end{equation*}

Assuming the independence of the \asum column $g(\cdot)$  and the oracle results 
$O(\cdot)$, we can further simplify the expression of MSE as follows.
\begin{align*}
& MSE_{\asum} = \frac{1}{b_2^{(s)}} \cdot \\
    & \left( 
    \mathbb{E}\left[g(s)^2\right]\sum_{i \notin \beta} 
    \left(|D_i| \mathbb{E}\left[\frac{O(s)}{r_i W(s)}\right]\right)
    - \frac{\mathbb{E}[g(s)]^2}{r_i} \sum_{i \notin \beta} 
    \left( |D_i| \mathbb{E}\left[O(s)\right] \right)^2 \right)
\end{align*}
We notice that $r_i$ unweights the proxy scores over a stratum into the proxy 
scores over the whole sampling regions. In this case, we merge $r_iW(s)$ as 
$W(s)$. Furthermore, we can simplify the sum of the expectations of strata into 
the expectation of the whole sampling region.
\begin{align*}
&MSE_{\asum}
= \frac{1}{b_2^{(s)}} \cdot \\
& \left( |D| \cdot \mathbb{E}\left[g(s)^2\right] \cdot 
    \mathbb{E}\left[\frac{O(s)}{W(s)}\right] - 
    \frac{\mathbb{E}[g(s)]^2}{r_i} \sum_{i \notin \beta} 
    \left(|D_i|\mathbb{E}\left[O(s)\right] \right)^2 \right)
\end{align*}
Next, we rewrite the expectations over all tuples with expectations over 
matching tuples by evaluating the oracle results $O(\cdot)$.
\begin{align}
& MSE_{\asum} \notag \\
&= \frac{1}{b_2^{(s)}} \left( 
    \mathbb{E}\left[g(s)^2\right] \cdot \left|\tilde{D}^{(s)}\right| \cdot 
    \mathbb{E}_{s\in \tilde{D}^{(s)}}\left[\frac1{W(t)}\right] 
        - \mathbb{E}[g(s)]^2 \sum_{i \notin \beta}  \frac{|\tilde{D}_i|^2}{r_i} \right) \notag \\
    & = \frac{\left|\tilde{D}^{(s)}\right|}{b_2^{(s)}} \left( 
        \mathbb{E}\left[g(s)^2\right]
        \mathbb{E}_{s\in \tilde{D}^{(s)}}\left[\frac1{W(t)}\right] 
        - \mathbb{E}[g(s)]^2 \sum_{i \notin \beta} q_i \left|\tilde{D}^{(s)}_i\right| \right) \notag
\end{align}
where
\begin{equation*}
    q_i = \frac{\left|\tilde{D}^{(s)}_i\right|}
        {\sum_{j \notin \beta} \left|\tilde{D}^{(s)}_j\right|} \bigg/ r_i
\end{equation*}

We rewrite the $MSE$ of \wwj for a \asum aggregate.
\begin{align*}
MSE_{\asum}^{(w)} & = \frac1b \left(
    \mathbb{E}_{s\sim W}\left[ \left( \frac{g(s)O(s)}{W(s)} \right)^2 \right] 
        - \left(\mathbb{E}_{s\sim W}\left[\frac{g(s)O(s)}{W(s)} \right] \right)^2 
        \right) \\
    & = \frac1b \left( 
        |D| \cdot \mathbb{E}\left[ \frac{g(s)^2 O(s)}{W(s)} \right] 
        - \left(|D| \cdot \mathbb{E}\left[ g(s)O(s) \right] \right)^2
    \right) \\
    & = \frac1b \left( 
        \mathbb{E}\left[ g(s)^2 \right] \cdot |D| \cdot 
        \mathbb{E}\left[\frac{O(s)}{W(s)}\right] 
        - |\tilde{D}|^2 \cdot \mathbb{E}[g(s)]^2 \right) \notag \\
    & = \frac{|\tilde{D}|}b \left( 
        \mathbb{E}\left[ g(s)^2 \right] 
        \mathbb{E}_{s\in \tilde{D}}\left[\frac1{W(t)}\right] 
        - \mathbb{E}[g(s)]^2|\tilde{D}| \right)
    \end{align*}
    
We take the ratio between the $MSE$ of \bas and that of \wwj. We assume the 
variance of \asum-column in the sampling region is the same as that for all 
tuples.
\begin{align*}
    & \frac{MSE_{\asum}}{MSE_{\asum}^{(w)}} \\
    & = \frac{|\tilde{D}^{(s)}|\big/ b_2^{(s)}}{|\tilde{D}|/b}
        \frac{ \mathbb{E}\left[g(s)^2\right]
            \mathbb{E}_{s\in \tilde{D}^{(s)}}\left[\frac1{W(s)}\right] 
            - \mathbb{E}[g(s)]^2 \sum_{i\notin \beta} q_i \left|\tilde{D}^{(s)}_i\right|}
        {\mathbb{E}\left[ g(s)^2 \right] 
            \mathbb{E}_{s\in \tilde{D}}\left[\frac1{W(s)}\right] 
            - \mathbb{E}[g(s)]^2|\tilde{D}|} \\
    & = \frac{|\tilde{D}^{(s)}|\big/ b_2^{(s)}}{|\tilde{D}|/b}
        \frac{ \mathbb{E}_{s\in \tilde{D}^{(s)}}\left[\frac1{W(s)}\right] 
            - \frac{\mathbb{E}[g(s)]^2}{\mathbb{E}[g(s)^2]} 
            \sum_{i \notin\beta} q_i \left|\tilde{D}^{(s)}_i\right|}
        {\mathbb{E}_{s\in \tilde{D}}\left[\frac1{P(s)}\right] 
            - \frac{\mathbb{E}[g(s)]^2}{\mathbb{E}[g(s)^2]}|\tilde{D}|} \\
    & = \frac{|\tilde{D}^{(s)}|\big/ b_2^{(s)}}{|\tilde{D}|/b}
        \frac{ \mathbb{E}_{s\in \tilde{D}^{(s)}}\left[\frac1{W(t)}\right] 
            - \frac{1}{1+CV_g} \sum_{i \notin \beta} q_i 
                \left|\tilde{D}^{(s)}_i\right|}
        {\mathbb{E}_{s\in \tilde{D}}\left[\frac1{W(s)}\right] 
            - \frac{1}{1+CV_g}|\tilde{D}|} \\
\end{align*}
where $CV_g$ is the coefficient of variation of the \asum-column $g(\cdot)$,
\begin{equation*}
CV_g = \frac{Var[g(s)]}{E[g(s)]^2}
\end{equation*}

We apply Holder's Inequality to obtain an upper bound to the effect of 
stratification on the MSE as follows.
\begin{align*}
&\sum_{i \notin \beta} q_i \left|\tilde{D}^{(s)}_i\right| 
= \frac{1}{|\tilde{D}^{(s)}|} \sum_{i \notin \beta} 
    \frac{\left|\tilde{D}^{(s)}_i\right|^2}{r_i}
= \frac{1}{|\tilde{D}^{(s)}|} \left(
    \sum_{i \notin \beta} \frac{\left|\tilde{D}^{(s)}_i\right|^2}{r_i}
    \right)\left(\sum_{i \notin \beta}r_i\right) \\
&\ge \frac{1}{|\tilde{D}^{(s)}|} \left(
    \sum_{i \notin \beta} \frac{\left|\tilde{D}^{(s)}_i\right|}{\sqrt{r_i}}\sqrt{r_i}
\right)^2 = \left|\tilde{D}^{(s)}\right| 
\end{align*}
Therefore, the ratio has the following upper bound.
\begin{align*}
\frac{MSE_{\asum}}{MSE_{\asum}^{(w)}} 
& \le \frac{|\tilde{D}^{(s)}|\big/ b_2^{(s)}}{|\tilde{D}|/b}
    \frac{ \mathbb{E}_{s\in \tilde{D}^{(s)}}\left[\frac1{W(s)}\right] 
        - \frac{1}{1+CV_g} |\tilde{D}^{(s)}|} 
    {\mathbb{E}_{s\in \tilde{D}}\left[\frac1{W(s)}\right] - \frac{1}{1+CV_g}|\tilde{D}|} \\
&= \frac{|\tilde{D}^{(s)}|^2\big/ b_2^{(s)}}{|\tilde{D}|^2/b} 
    \frac{ \mathbb{E}_{s\in \tilde{D}^{(s)}}
        \left[\frac{1/|\tilde{D}^{(s)}|}{W(s)}\right] - \frac{1}{1+CV_g}}
    {\mathbb{E}_{s\in \tilde{D}}\left[\frac{1/|\tilde{D}|}{W(s)}\right] 
    - \frac{1}{1+CV_g}} 
\end{align*}
If the following condition holds
\begin{equation}
\mathbb{E}_{s\in \tilde{D}^{(s)}}\left[\frac{1/|\tilde{D}^{(s)}|}{W(s)}\right] 
\le \mathbb{E}_{s\in \tilde{D}}\left[\frac{1/|\tilde{D}|}{W(s)}\right], 
\label{eq:supp-cond1}
\end{equation}
we have the following upper bound for the ratio
\begin{align*}
\frac{MSE_{\asum}}{MSE_{\asum}^{(IS)}} 
< \frac{|\tilde{D}^{(s)}|^2\big/ b_2^{(s)}}{|\tilde{D}|^2/b} 
    \frac{ \mathbb{E}_{s\in \tilde{D}^{(s)}}
        \left[\frac{1/|\tilde{D}^{(s)}|}{W(s)}\right]}
    {\mathbb{E}_{s\in \tilde{D}}\left[\frac{1/|\tilde{D}|}{W(s)}\right]}
\end{align*}
Furthermore, if the matching tuples are sparse in the sampling region, \ie,
\begin{equation}
    |\tilde{D}^{(s)}|^2\big/ b_2^{(s)} < |\tilde{D}|^2/b \label{eq:supp-cond2}
\end{equation}
Then,
\begin{equation*}
    MSE_{\asum} = C \cdot {MSE_{\asum}^{(IS)}}
\end{equation*}
where $C$ is a coefficient less than 1,
\begin{equation*}
    C < \frac{ \mathbb{E}_{s\in \tilde{D}^{(s)}}\left[\frac{1/|\tilde{D}^{(s)}|}{W(s)}\right]}{
        \mathbb{E}_{s\in \tilde{D}}\left[\frac{1/|\tilde{D}|}{W(s)}\right]} \le 1
\end{equation*}

We then apply the Theorem \ref{theom:opt} to replace the MSE of \bas with 
deterministic allocation with the MSE of \bas with pilot sampling. Namely, we 
have the following approximation:
\begin{align*}
    MSE_{\asum}(\hat\beta^*) &= MSE_{\asum}(\beta^*) \cdot 
        \left(1+\mathcal{O}\left(b_1^{-1/2}\right)\right) \\
        &= C \cdot {MSE_{\asum}^{(w)}}\left(1+\mathcal{O}\left(b_1^{-1/2}\right)\right) \notag
\end{align*}
Since $MSE_{\asum}^{(w)}$ converges to 0 at the rate 
$\mathcal{O}\left(b^{-1}\right)$. Therefore, we conclude that if the conditions 
\ref{eq:supp-cond1} and \ref{eq:supp-cond2} hold, \name outperforms \wwj 
asymptotically. Namely,
\begin{equation}
    MSE_{\asum}(\hat\beta^*) = C \cdot {MSE_{\asum}^{(IS)}} + 
    \mathcal{O}\left(b_1^{-1/2}b^{-1}\right) \label{eq:outperform-supp}
\end{equation}
On the other hand, if the conditions \ref{eq:supp-cond1} and \ref{eq:supp-cond2} 
do not hold, we can set $\beta=\emptyset$, which will make the sampling region 
become the entire data tuples. In this case, the upper bound of the ratio will 
be 1. Namely,
\begin{equation*}
    MSE_{\asum} \le MSE^{(IS)}_{\asum}
\end{equation*}
Taking Theorem \ref{theom:opt} into account, we will have the following upper 
bound for MSE of \bas.
\begin{align}
    MSE_{\asum} &\le MSE_{\asum} \left(1+\mathcal{O}\left(b_1^{-1/2}\right)\right) 
    \le MSE^{(w)}_{\asum}\left(1+\mathcal{O}\left(b_1^{-1/2}\right)\right) \notag \\
    &= MSE^{(w)}_{\asum} + \mathcal{O}\left(b_1^{-1/2}b^{-1}\right) \label{eq:match-supp}
\end{align}

To conclude, we have shown that the MSE of \bas either outperforms (Eq. 
\ref{eq:outperform-supp}) or matches (Eq. \ref{eq:match-supp}) that of 
\wwj asymptotically. 
\end{proof}

%% file: tex/appendix/selection.tex
\subsection{\bas for Selection Join Queries} \label{sec:app-select}
\begin{lemma}
    With a probability higher than $p$, we can achieve the overall recall target 
    $\gamma$ if $\gamma_s$ satisfies
    \begin{equation*}
        \gamma_s \ge \gamma - (1-\gamma)\frac{\acount_b}{\mathrm{UB}(\acount_s, Var[\acount_s], b, p)}
    \end{equation*}
    where
    \begin{equation*}
        \mathrm{UB}(\mu, \sigma^2, b, p) = \mu + \frac{\sigma}{\sqrt{b}}\sqrt{2\log \frac{2}{1-p}}
    \end{equation*}
    \label{lemma:select}
    \end{lemma}
\begin{proof}
We first show the upper bound of the number of matching tuples in the sampling 
region $\widehat{\acount}_s$. Then, the required recall target $\gamma_s$ of the 
sampling region follows automatically.

Given an i.i.d. sample of size $n$ drawn from a population with mean $\mu$ and 
finite and non-zero variance $\sigma^2$, the upper bound of the sample mean can 
be estimated with normal approximation \cite{supg,aqp-normal-approx}. Namely, 
\begin{equation*}
    \mathbb{P}\left[\hat\mu \ge  \mu + \frac{\sigma}{\sqrt{n}}\sqrt{2\log \frac{2}{1-p}}\right] \le \frac{1-p}{2}
\end{equation*}
where $p$ is the confidence.

We rewrite the recall target of the sampling region $\gamma'_s$ with the overall 
recall target $\gamma$ specified by the user.
\begin{equation*}
    \gamma_s = 
    \frac{\gamma\left(\widehat{\acount}_s+\widehat{\acount}_b\right) 
        -\widehat{\acount}_b}
    {\widehat{\acount}_s} = \gamma - (1-\gamma)\frac{\widehat{\acount}_b}{\widehat{\acount}_s}
\end{equation*}
We observe that $\gamma_s$ is monotonically increasing with respect to 
$\widehat{\acount}_s$. Therefore, we use the upper bound of $\widehat{\acount}_s$
to estimate the required $\gamma_s$ such that the user-specified overall recall 
target can be guaranteed with high probability.
\end{proof}

%% file: tex/appendix/bootstrap.tex
\subsection{Bootstrap Validity} \label{sec:app-bootstrap}
In this section, we show that bootstrap resampling is asymptotically valid for
\bas, which is inspired by \cite{abae}.

\minihead{Overview} We summarize our approach as follows. Suppose that the sampling
procedure with CDF $F_0$ leads to a sequence of random variables 
$X_1, \ldots, X_N$, and $\theta$ is a statistical functional of $F_0$. If 
$\theta$ is Hadamard differentiable, the delta method of \cite{asymptotic} will
imply the validity of using bootstrap resampling for the CIs of $\theta$. In the
rest of this section, we first define $\theta$ for \bas and derive the Hadamard
differentiability of $\theta$.

\minihead{Notation} We first define the notations for the derivation. Suppose 
$X_i = (P_i, Y_i)$ is a tuple drawn from the cross product of tables, 
where $P_i \in \{0, 1\}$ indicates whether the tuple satisfy the join condition,
$Y_i$ is the value to be \asum-ed or \aavg-ed. We can define the 
statistical function $\theta$ for \acount, \asum, and \aavg as follows:
\begin{align*}
    &\theta_{f,\acount}  = \sum_{i} \frac{P_i}{N \cdot f(X_i)}, \quad
    \theta_{f,\asum} = \sum_{i} \frac{P_i \cdot Y_i}{N \cdot f(X_i)}, \\
    &\theta_{f,\aavg} = \frac{\sum_{i} \frac{P_i \cdot Y_i}{N \cdot f(X_i)}}{\sum_{i} \frac{P_i}{N \cdot f(X_i)}} 
\end{align*}
where $f(\cdot)$ is a deterministic distribution denoting the sampling probability.

\begin{theorem}
Let $\dot{\theta}_{f,F_0,\texttt{AGG}}\left(\mathbb{B}\right)$ be the converged 
distribution of the difference between the empirical and statistical distribution:
\begin{equation*}
    \sqrt{b}\left(\theta_{f, \texttt{AGG}}\left(\hat F_b\right) - 
    \theta_{f, \texttt{AGG}}\left(F_0\right)\right)
    \xrightarrow{\mathcal{D}}
    \dot{\theta}_{f,F_0,\texttt{AGG}}\left(\mathbb{B}\right)
\end{equation*}
where $b$ is the sample size, $F_0$ is the CDF of the population, and $\hat F_b$
is the empirical CDF of the observed sample.
For any fixed sampling probability $f$ and aggregates 
$\texttt{AGG} \in \{\acount, \asum, \aavg\}$, we have the following convergence 
in distribution:
\begin{equation*}
    \sqrt{b} \left(\theta_{f, \texttt{AGG}}\left(\hat F_b^*\right) - 
    \theta_{f, \texttt{AGG}}\left(\hat F_b\right)\right) \xrightarrow{\mathcal{D}}
    \dot{\theta}_{f,F_0,\texttt{AGG}}\left(\mathbb{B}\right)
\end{equation*}
where $\hat F_b^*$ is the resampled CDF. Namely, the bootstrap is valid.
\end{theorem}

\noindent\begin{proof}
According to Theorem 23.9 of \cite{asymptotic}, it is sufficient to show that
$\theta_{f, \texttt{AGG}}$ is Hadamard differentiable.

For any $f$ with $\sum_{i}\frac{P_i}{f(X_i)} \ge 0$, $\theta_{f, \acount}$ is 
Hadamard differentiable. For any $f$ with 
$\sum_{i}\frac{P_i \cdot Y_i}{f(X_i)} \ge 0$, $\theta_{f, \asum}$ is also 
Hadamard differentiable. For any join with non-empty results, 
$\sum_{i} \frac{P_i}{N \cdot f(X_i)}$ is non-zero. Then, we apply the chain rule 
(Theorem 20.9, \cite{asymptotic}) and show that $\theta_{f, \acount}^{-1}$ is 
Hadamard differentiable. We then apply the product rule, implying that 
$\theta_{f, \aavg}$ is Hadamard differentiable.

\end{proof}

%% file: tex/appendix/query_semantics.tex
\begin{table*}[t!]
    \centering
    \begin{tabular}{lllllll}
        \toprule
        Dataset       & Modality    & Table size                           & Cross product size   & Positive rate        & Oracle      & Embedding model \\ %
        \midrule
        \dcompany     & text        & $1.0 \times 10^4$, $1.0 \times 10^4$ & $1.0 \times 10^8$    & $3.6 \times 10^{-5}$ & Human label & MiniLM \cite{model-minilm}      \\ %
        \dquora       & text        & $6.0 \times 10^4$                    & $3.6 \times 10^9$    & $1.1 \times 10^{-6}$ & Human label & MiniLM \cite{model-minilm}      \\ %
        \dwebmasters  & text        & $1.8 \times 10^4$                    & $3.2 \times 10^8$    & $4.6 \times 10^{-5}$ & Human label & MiniLM \cite{model-minilm}      \\ %
        \droxford     & image       & $70$, $1.0 \times 10^6$ & $7.0 \times 10^7$    & $7.8 \times 10^{-5}$ & Human label & SuperGlobal \cite{shao2023global}      \\ %
        \dveri        & image       & $2.1 \times 10^4$, $2.8 \times 10^4$ & $6.0 \times 10^8$    & $1.4 \times 10^{-3}$ & Human label & TransReID \cite{model-transreid}  \\ %
        \dflickr      & text, image & $3.2 \times 10^4$, $1.6 \times 10^5$ & $5.1 \times 10^9$    & $3.1 \times 10^{-5}$ & Human label & BLIP \cite{model-blip}       \\ %
        Ecomm-Q7      & text        & $8.1 \times 10^3$                    & $6.5 \times 10^7$    & $8.1 \times 10^{-3}$ & Human label & MiniLM \cite{model-minilm} \\
        Ecomm-Q8      & text, image & $3.0 \times 10^2$, $4.4 \times 10^4$ & $1.3 \times 10^7$    & $2.3 \times 10^{-5}$ & Human label & CLIP \cite{clip} \\
        Ecomm-Q9      & image       & $6.1 \times 10^3$                    & $3.7 \times 10^7$    & $2.4 \times 10^{-2}$ & Human label & CLIP \cite{clip} \\
        Ecomm-Q10     & text, image & $4.6 \times 10^2$, $9.0 \times 10^2$, $4.2 \times 10^3$ & $1.7 \times 10^9$ & $2.1 \times 10^{-5}$ & Human label & CLIP \cite{clip} \\
        Ecomm-Q11     & text, image & $370$, $600$, $2.0 \times 10^3$, $2.9 \times 10^3$ & $1.3 \times 10^{12}$ & $5.6 \times 10^{-6}$ & Human label & CLIP \cite{clip} \\
        Movie-Q5,Q6   & text        & $7.6 \times 10^2$                    & $5.7 \times 10^5$    & $5.0 \times 10^{-1}$   & Human label & Flair \cite{akbik2019flair} \\
        \dcompany-Scale & text & $1.0 \times 10^4, 3 \times 10^3, 300, 25, 25, 20$ & $1.1 \times 10^{14}$ & $3.0 \times 10^{-10}$ & Human label & MiniLM \cite{model-minilm} \\
        \droxford-Scale & image & $70, 1.0 \times 10^{7}$ & $7 \times 10^{8}$ & $7.8 \times 10^{-5}$ & Human label & SuperGlobal \cite{shao2023global} \\
        Syn           & N/A         & $1.0 \times 10^4$, $1.0 \times 10^4$ & $1.0 \times 10^8$    & $1 \times 10^{-4}$ & Synthetic & Synthetic \cite{supg}      \\ %
        \bottomrule
    \end{tabular}
    \caption{Summary of datasets, Oracle, and embedding models.}
    \label{tab:datasets}
\end{table*}

\section{Datasets and Queries}

In this section, we introduce the construction details of the datasets and 
queries used in evaluation.

\subsection{Dataset Collection and Construction}
\label{sec:app-dataset}

\minihead{Quora} 
Widely used to evaluate the sentence paraphrasing task in the NLP community, the 
Quora dataset contains single-sentence questions from the forum quora.com, 
labeled by human \cite{dataset-quora,dataset-quora-work}. We repurpose it to a 
paraphrasing detection task with the self-join operation.

\minihead{WebMaster} 
Collected for knowledge mining tasks and studied for duplication detection tasks, 
the WebMaster dataset contains long posts from the network question partition of 
the QA forum stackexchange.com for developers, labeled by human 
\cite{dataset-stackoverflow,dataset-stackoverflow-work1}. We repurpose it to a 
duplication detection task.

\minihead{Company} 
As widely used to evaluate the EM systems, the Company dataset has two lists of 
long descriptions about companies, where the matched companies between the lists 
are labeled by human \cite{dataset-company-work1,dataset-company-work2}. We use 
it as a two-table join task. Additionally, we split the two tables in six tables
with sizes: 20, 25, 25, 300, 3,000, and 10,000, simulating the join ``barebones'' 
of TPC-H Q7---the largest join evaluated by Wander Join \cite{structured-join-index}.

\minihead{VeRi}
Proposed to evaluate the vehicle re-identification algorithms, the VeRi dataset 
contains images of vehicles captured by multiple cameras on different crossroads, 
where the identical vehicles are labeled by human \cite{veri1,veri2}. We 
repurpose it to a two-table join task.

\minihead{Flickr30K} 
Proposed to evaluate the image-to-sentence models, the Flickr30K dataset 
contains the image entities and sentence mentions from the Flickr.com, where 
the matched images and sentences are labeled by human \cite{dataset-flickr30k}. 
We repurpose it to a two-table join task.

\minihead{Roxford}
Proposed to evaluate object retrieval and visual search algorithms, the Roxford 
(Oxford Buildings) dataset consists of images depicting famous architectural 
landmarks in Oxford collected from Flickr, where images representing the same 
building are labeled by human \cite{radenovic2018revisiting}. We repurpose it to 
a two-table join task for visual landmark identification. Additionally, we 
uniformly scale up the dataset to construct a query involving 10 million rows.

\minihead{SemBench}
As a benchmark for semantic query processing engines, SemBench contains semantic
queries from six application domains. From the queries involving join operations,
we collected ones that have at least one join predicate about the joint semantics
across tables. When there are semantic operators other than join (e.g., filtering),
we preprocess these in order to isolate the evaluation of joins. Additionally, 
we introduced two modified versions of queries to cover \agroupby and \atopk 
queries.

\minihead{Syn}
Designed to stress-test join algorithms under varying noise levels, this dataset 
consists of two tables of size 10,000 with synthetic similarity scores. We 
assign ground-truth matches randomly based on selectivity and sample scores from 
distinct Beta distributions to distinguish matches ($Beta(5, 0.5)$) from 
non-matches ($Beta(0.5, 5)$), following prior work \cite{supg}. We subsequently 
inject noise by inverting the score distributions for subsets of pairs to 
simulate precise False Positive and False Negative rates.

\subsection{Query Semantics}
\label{sec:app-query}

We define the semantics of the evaluated queries as follows, utilizing domain-specific semantic predicates.

\minihead{Company (Entity Resolution)} Estimating the number of overlapping entities between two distinct corporate data sources (e.g., Wikipedia vs. DBPedia) based on textual descriptions. 
\begin{lstlisting}[keywords={SELECT,AVG,FROM,JOIN,ON,NL,WHERE}]
SELECT COUNT(*) 
FROM wiki_companies AS c1 
JOIN dbpedia_companies AS c2 
ON NL(`Company 1: {c1.description} and Company 2: {c2.description} are the same company.') 
\end{lstlisting}

\minihead{Quora (Paraphrase Detection)} Analyzing data quality by computing the average difference in character length between questions marked as semantic duplicates. 
\begin{lstlisting}[keywords={SELECT,AVG,FROM,JOIN,ON,NL,WHERE}]
SELECT AVG(ABS(q1.char_len - q2.char_len)) 
FROM quora_questions AS q1 
JOIN quora_questions AS q2 
ON NL('Question 1: {q1.text} and Question 2: {q2.text} are paraphrased from each other.') 
WHERE q1.id < q2.id 
\end{lstlisting}

\minihead{Webmasters (Redundancy Impact)} 
Calculating the total community effort (number of answers) wasted on duplicate forum posts.
\begin{lstlisting}[keywords={SELECT,AVG,FROM,JOIN,ON,NL,WHERE}]
SELECT SUM(p1.answer_count) 
FROM forum_posts AS p1 
JOIN forum_posts AS p2 
ON NL('Post 1: {p1.post} and Post 2: {p2.post} are duplicated.')
WHERE p1.id != p2.id 
\end{lstlisting}

\minihead{Roxford (Visual Search)} 
Counting the number of images in a gallery that depict the same building/landmark as the query images. 
\begin{lstlisting}[keywords={SELECT,AVG,FROM,JOIN,ON,NL}]
SELECT COUNT(*) 
FROM oxbuild_gallery AS g 
JOIN oxbuild_queries AS q 
ON NL('Two images show the same landmark: {g.image} {q.image}') 
\end{lstlisting}

\minihead{Flickr-30K (Cross-Modal Retrieval)} 
Measuring the alignment between an image dataset and a caption dataset by counting valid image-text pairs. 
\begin{lstlisting}[keywords={SELECT,AVG,FROM,JOIN,ON,NL}]
SELECT COUNT(*) 
FROM flickr_images AS img 
JOIN flickr_captions AS cap 
ON NL('Caption {cap.text} describes the image {img.image}') 
\end{lstlisting}

\minihead{VeRi (Traffic Analytics)} 
Calculating the average transit time for vehicles traveling between two different camera locations (Re-ID). 
\begin{lstlisting}[keywords={SELECT,AVG,FROM,JOIN,ON,NL}]
SELECT AVG(cam2.timestamp - cam1.timestamp) 
FROM surveillance_cam_A AS cam1 
JOIN surveillance_cam_B AS cam2 
ON NL('Two images show the same vehicle: {cam1.image} {cam2.image}') 
\end{lstlisting}

\minihead{Ecomm-Q7}
Counting matching pairs for budget-friendly fashion items based on strict brand and category alignment. 
\begin{lstlisting}[keywords={SELECT,AVG,FROM,JOIN,ON,NL,WHERE}] 
SELECT COUNT(*) 
FROM fashion_styles AS t1 
JOIN fashion_styles AS t2 
ON NL('Product 1 {t1.text} and Product 2 {t2.text} have the same brand and type.')
WHERE t1.price <= 500 AND t2.price <= 500 
\end{lstlisting}

\minihead{Ecomm-Q8}
Calculating the minimum price of the product which has an image and have extensive
textual descriptions (exceeding 3,000 characters).
\begin{lstlisting}[keywords={SELECT,AVG,FROM,JOIN,ON,NL,WHERE}] 
SELECT t.id, i.id 
FROM fashion_text AS t 
JOIN fashion_images AS i 
ON NL('The product description: {t.text} describes the image {i.image}')
WHERE LENGTH(t.description) >= 3000 
\end{lstlisting}

\minihead{Ecomm-Q9}
Calculating the median price difference of two budget-friendly (< 800), 
single-colored fashion items that share identical primary color and category. 
\begin{lstlisting}[keywords={SELECT,AVG,FROM,JOIN,ON,NL,WHERE,AND,IS,NULL,IN}] 
SELECT MEDIAN(ABC(t1.price - t2.price))
FROM fashion_items AS t1 
JOIN fashion_items AS t2 
ON NL('Two images show the products with the same color and category: {t1.image} {t2.image}')
WHERE t1.price < 800 AND t1.baseColour IN ('Black', 'Blue', 'Red', 'White', 'Orange', 'Green')
\end{lstlisting}

\minihead{Ecomm-Q10}
Counting the number of coherent three-piece outfits (Footwear, Bottomwear, 
Topwear) by linking budget-friendly items (< 1000) that share the same brand and 
base color. 
\begin{lstlisting}[keywords={SELECT,FROM,JOIN,ON,WHERE,AND,IN}]
SELECT COUNT(*)
FROM footwear AS s 
JOIN bottomwear AS b 
ON NL('Two images show the products with the same brand and color: {s.image} {b.image}')
JOIN topwear AS t 
ON NL('Two images show the products with the same brand and color: {b.image} {t.image}')
WHERE s.price <= 1000 AND s.baseColour IN ('Black', 'Blue', 'Red', 'White') 
\end{lstlisting}

\minihead{Ecomm-Q10a}
Counting the number of coherent three-piece outfits (Footwear, Bottomwear, 
Topwear) by linking budget-friendly items (< 1000) that share the same brand and 
base color and grouping the results by color.
\begin{lstlisting}[keywords={SELECT,FROM,JOIN,ON,WHERE,AND,IN,GROUP,BY}]
SELECT COUNT(*)
FROM footwear AS s 
JOIN bottomwear AS b 
ON NL('Two images show the products with the same brand and color: {s.image} {b.image}')
JOIN topwear AS t 
ON NL('Two images show the products with the same brand and color: {b.image} {t.image}')
WHERE s.price <= 1000 AND s.baseColour IN ('Black', 'Blue', 'Red', 'White') 
GROUP BY s.baseColour
\end{lstlisting}

\minihead{Ecomm-Q11}
Counting complete four-piece outfits (Footwear, Bottomwear, Topwear, Accessories) 
by aligning brands, requiring a monochromatic black aesthetic for the clothing 
and a budget constraint ($< 500$) for the accessories.
\begin{lstlisting}[keywords={SELECT,FROM,JOIN,ON,WHERE,AND,IN}]
SELECT COUNT(*)
FROM footwear AS s
JOIN bottomwear AS b 
ON NL('Two images show the products with the same brand: {s.image} {b.image}')
JOIN topwear AS t 
ON NL('Two images show the products with the same brand: {b.image} {t.image}')
JOIN accessories AS a 
ON NL('Two images show the products with the same brand: {t.image} {a.image}')
WHERE s.baseColour = 'Black'
AND b.baseColour = 'Black'
AND t.baseColour = 'Black'
AND a.price <= 500
AND a.subCategory IN ('Bags', 'Jewellery', 'Watches')
\end{lstlisting}

\minihead{Ecomm-Q11a}
Find the Top-5 most frequent fabrics used in $<\mathdollar500$ topwear cloth 
that can pair with footwear, bottomwear, and accessories (Bags, Jewellerys, or
Watches) from the same brand.
\begin{lstlisting}[keywords={SELECT,FROM,JOIN,ON,WHERE,AND,IN,GROUP,BY,ORDER,LIMIT}]
SELECT t.fabric
FROM footwear AS s
JOIN bottomwear AS b 
ON NL('Two images show the products with the same brand: {s.image} {b.image}')
JOIN topwear AS t 
ON NL('Two images show the products with the same brand: {b.image} {t.image}')
JOIN accessories AS a 
ON NL('Two images show the products with the same brand: {t.image} {a.image}')
WHERE s.baseColour = 'Black'
AND b.baseColour = 'Black'
AND t.baseColour = 'Black'
AND a.price <= 500
AND a.subCategory IN ('Bags', 'Jewellery', 'Watches')
GROUP BY t.fabric
ORDER BY COUNT(*) 
LIMIT 5
\end{lstlisting}

\minihead{Movie-Q5}
Computing the average score sum of each pair of reviews for the movie ``Ant-Man 
and the Wasp: Quantumania'' that share the same sentiment classification. 
\begin{lstlisting}[keywords={SELECT,FROM,JOIN,ON,WHERE,AND}]
SELECT AVG(r1.score + r2.score) 
FROM movie_reviews AS r1 JOIN movie_reviews AS r2 
ON NL('Review 1 {r1.text} and Review 2 {r2.text} share the same semtiment')
WHERE r1.movie_id = 'ant_man_and_the_wasp_quantumania' 
\end{lstlisting}

\minihead{Movie-Q6}
Computing the maximum score sum of each pair of reviews for the movie ``Ant-Man 
and the Wasp: Quantumania'' that have the opposite sentiment classification. 
\begin{lstlisting}[keywords={SELECT,FROM,JOIN,ON,WHERE,AND}]
SELECT MAX(r1.score + r2.score) 
FROM movie_reviews AS r1 JOIN movie_reviews AS r2 
ON NL('Review 1 {r1.text} and Review 2 {r2.text} share the opposite semtiment')
WHERE r1.movie_id = 'ant_man_and_the_wasp_quantumania' 
\end{lstlisting}